\newcommand{\arxiv}[1]{#1}

\documentclass[a4paper,UKenglish,cleveref,autoref,thm-restate]{lipics-v2021}

\arxiv{\pdfoutput=1}
\arxiv{\hideLIPIcs}

\usepackage{graphicx}
\usepackage{esvect} 
\usepackage{tcolorbox}
\usepackage{complexity}
\usepackage{tikz}
\usepackage{amsmath}
\usepackage{mathrsfs}
\usepackage{multibib}
\usepackage{enumerate}
\usepackage{algorithmic}
\usepackage{booktabs}
\usepackage{multirow}
\usepackage{lineno}
\usepackage{tabularx}
\usepackage{arydshln}
\usepackage{hyperref}
\usepackage{caption}
\usepackage[T1]{fontenc}
\usepackage{lmodern}
\usepackage{mathtools}
\usepackage[bibliography=common]{apxproof}

\usetikzlibrary{decorations.markings}

\newcommand{\prb}[1]{\textnormal{\textsc{#1}}}
\newcommand{\CS}{\mathsf{CS}}
\newcommand{\CRCS}{{\prb{CRCS}}}
\newcommand{\kCRCS}{{\prb{$k$-CRCS}}}
\newcommand{\threeCRCS}{{\prb{$3$-CRCS}}}
\newcommand{\ECRCS}{\prb{ECRCS}}

\newcommand{\kCR}{\prb{$k$-Coloring Reconfiguration}}

\newcommand{\invariant}{invariant}
\newcommand{\inv}{\mathop{\mathsf{inv}}}
\newcommand{\symdif}{\mathbin{\triangle}}
\newcommand{\swappable}[1]{{S}_{#1}}
\newcommand{\portv}[2]{{#1}_{#2}} 

\newcommand{\prbNCL}{\prb{Nondeterministic Constraint Logic}}

\newtheorem{redrule}{Reduction Rule}
\Crefname{redrule}{Rule}{Rules}

\newcommand{\TS}{token sliding}

\newcommand{\rev}[1]{{#1}}

\title{Coloring Reconfiguration under Color Swapping} 

\author{Janosch Fuchs}{IT Center, RWTH Aachen University, Germany}{fuchs@itc.rwth-aachen.de}{https://orcid.org/0000-0003-3993-222X}{}

\author{Rin Saito}{Graduate School of Information Sciences, Tohoku University, Japan}{rin.saito@dc.tohoku.ac.jp}{https://orcid.org/0000-0002-3953-4339}{Supported by JST SPRING Grant Number JPMJSP2114.}

\author{Tatsuhiro Suga}{Graduate School of Information Sciences, Tohoku University, Japan}{suga.tatsuhiro.p5@dc.tohoku.ac.jp}{https://orcid.org/0009-0002-1376-4678}{}

\author{Takahiro Suzuki}{Graduate School of Information Sciences, Tohoku University, Japan}{takahiro.suzuki.q4@dc.tohoku.ac.jp}{https://orcid.org/0009-0005-8433-3789}{}

\author{Yuma Tamura}{Graduate School of Information Sciences, Tohoku University, Japan}{tamura@tohoku.ac.jp}{https://orcid.org/0009-0001-5479-7006}{Supported by JSPS KAKENHI Grant Number JP25K21148.}

\authorrunning{J. Fuchs, R. Saito, T. Suga, T. Suzuki, and Y. Tamura} 

\Copyright{Janosch Fuchs, Rin Saito, Tatsuhiro Suga, Takahiro Suzuki, and Yuma Tamura} 



\category{} 

\relatedversion{} 

\nolinenumbers 

\EventEditors{John Q. Open and Joan R. Access}
\EventNoEds{2}
\EventLongTitle{42nd Conference on Very Important Topics (CVIT 2016)}
\EventShortTitle{CVIT 2016}
\EventAcronym{CVIT}
\EventYear{2016}
\EventDate{December 24--27, 2016}
\EventLocation{Little Whinging, United Kingdom}
\EventLogo{}
\SeriesVolume{42}
\ArticleNo{23}

\begin{document}

\maketitle

\begin{abstract}
	In the \textsc{Coloring Reconfiguration} problem, we are given two proper $k$-colorings of a graph and asked to decide whether one can be transformed into the other by repeatedly applying a specified recoloring rule, while maintaining a proper coloring throughout.
	For this problem, two recoloring rules have been widely studied: \emph{single-vertex recoloring} and \emph{Kempe chain recoloring}.
	In this paper, we introduce a new rule, called \emph{color swapping}, where two adjacent vertices may exchange their colors, so that the resulting coloring remains proper, and study the computational complexity of the problem under this rule.
	We first establish a complexity dichotomy with respect to $k$: the problem is solvable in polynomial time for $k \leq 2$, and is $\PSPACE$-complete for $k \geq 3$.
	We further show that the problem remains $\PSPACE$-complete even on restricted graph classes, including bipartite graphs, split graphs, and planar graphs of bounded degree.
	In contrast, we present polynomial-time algorithms for several graph classes: for paths when $k = 3$, for split graphs when $k$ is fixed, and for cographs when $k$ is arbitrary.
\end{abstract}

\section{Introduction}
\subsection{Reconfiguring of Colorings}
\label{subsec:intro_coloring_reconf}

The field of \emph{combinatorial reconfiguration} investigates reachability and connectivity in the solution space of combinatorial problems.
A {combinatorial} reconfiguration problem is defined with respect to a combinatorial problem~$\Pi$ and a \emph{reconfiguration rule}~$\textsf{R}$, which specifies how one feasible solution of~$\Pi$ can be transformed into another.
Given an instance of~$\Pi$ and two feasible solutions, the reconfiguration problem asks whether it is possible to transform one into the other through a sequence of solutions, each obtained by a single application of~$\textsf{R}$, such that all intermediate solutions are also feasible.
Combinatorial reconfiguration problems naturally arise in applications involving dynamic systems, where solutions must be updated while maintaining feasibility at every step.
In addition, studying such problems provides deeper insights into the structural properties of the solution spaces of classical combinatorial problems.
We refer the reader to the surveys by van den Heuvel~\cite{survey:Heuvel13} and Nishimura~\cite{survey:Nishimura18} for comprehensive overviews of this growing area.

Among the various reconfiguration problems, one of the most fundamental and extensively studied is the reconfiguration of graph colorings.
Let $k$ be a positive integer.
For a fixed reconfiguration rule~$\textsf{R}$, the \prb{Coloring Reconfiguration} problem under~$\textsf{R}$ is defined as follows:
given a $k$-colorable graph~$G$ and two proper $k$-colorings~$f$ and~$f'$, determine whether there exists a sequence of proper $k$-colorings of~$G$ starting at~$f$ and ending at~$f'$, where each coloring in the sequence is obtained from the previous one by a single application of~$\textsf{R}$.
The \prb{$k$-Coloring Reconfiguration} is a special case where $k$ is fixed.
The computational complexity of {\kCR} has been the subject of extensive algorithmic study; see Section~3 of the survey by Mynhardt and Nasserasr~\cite{survey:MynhardtN20}.

Two reconfiguration rules have been widely studied for {\kCR}: \emph{single-vertex recoloring} and \emph{Kempe chain recoloring}.
In the \emph{single-vertex recoloring} rule, a new proper $k$-coloring is obtained by recoloring a single vertex so that the resulting $k$-coloring remains proper.
Under this rule, {\kCR} is solvable in polynomial time when $k \leq 3$~\cite{single_recolor:CerecedaHJ11}, while it becomes $\PSPACE$-complete for every fixed $k \geq 4$~\cite{single_recolor:BonsmaC09}.
Notably, this rule is closely related to Glauber dynamics in statistical physics, where a Markov chain is defined over the space of proper $k$-colorings of a graph~$G$: at each step, a vertex is selected uniformly at random and recolored with a randomly chosen color such that the resulting coloring remains proper.
See Sokal~\cite{survey:Sokal05} for an introduction to the Potts model and its connections to graph coloring.

The second widely studied rule is \emph{Kempe chain recoloring}.
A new proper $k$-coloring is obtained by selecting a connected component~$C$ of the subgraph of~$G$ induced by two color classes (i.e., a \emph{Kempe chain}) and swapping the two colors within~$C$.
Note that when~$C$ consists of a single vertex, this operation is equivalent to single-vertex recoloring.
While {\kCR} under this rule is solvable in polynomial time when $k \leq 2$~\cite{Kempe:Mohar07} (in fact, the answer is always \rev{\texttt{YES}}), it becomes $\PSPACE$-complete for every fixed $k \geq 3$~\cite{Kempe:BonamyHIKMMSW20}.
Kempe chain recoloring was originally introduced by Kempe in 1879 in an attempt to prove the Four Color Theorem~\cite{Kempe:Kempe1879}.
Although his proof was later found to be flawed, the technique has continued to play a central role in graph coloring theory~\cite{Kempe:BonamyBFJ19,Kempe:Mohar07}, statistical physics~\cite{Mohar_2009,Mohar_2010}, and the study of mixing times of Markov chains~\cite{kempe:Vigoda99}.

These results highlight a key aspect of reconfiguration problems:
even when the definition of feasible solutions remains the same, the computational complexity of finding a reconfiguration sequence can vary greatly depending on the selected reconfiguration rule.

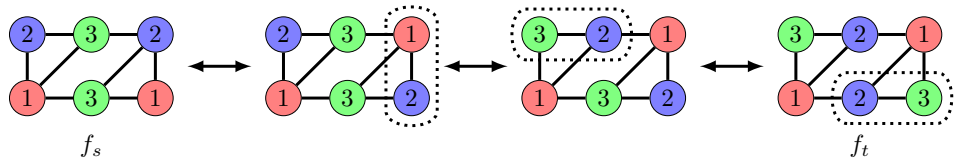
\begin{figure}
	\centering
	\resizebox{0.9\textwidth}{!}{\begin{tikzpicture}[scale=0.9]

\node[fill=red!50, draw=black, circle, minimum size=5mm, inner sep=0.5pt] (G1) at (0,0) {$1$};
\node[fill=blue!50, draw=black, circle, minimum size=5mm,  inner sep=0.5pt] (G2) at (0,1) {$2$};
\node[fill=green!50, draw=black, circle, minimum size=5mm, inner sep=0.5pt] (G3) at (1,0) {$3$};
\node[fill=green!50, draw=black, circle, minimum size=5mm, inner sep=0.5pt] (G4) at (1,1) {$3$};
\node[fill=red!50, draw=black, circle, minimum size=5mm,  inner sep=0.5pt] (G5) at (2,0) {$1$};
\node[fill=blue!50, draw=black, circle, minimum size=5mm, inner sep=0.5pt] (G6) at (2,1) {$2$};

\draw[very thick] (G1)--(G2)--(G4)--(G6)--(G5)--(G3)--(G1);
\draw[very thick] (G1)--(G4);
\draw[very thick] (G3)--(G6);

\node[] at (1,-0.75) {$f_s$};

\node[fill=red!50, draw=black, circle, minimum size=5mm, inner sep=0.5pt] (G1) at (4+0,0) {$1$};
\node[fill=blue!50, draw=black, circle, minimum size=5mm,  inner sep=0.5pt] (G2) at (4+0,1) {$2$};
\node[fill=green!50, draw=black, circle, minimum size=5mm, inner sep=0.5pt] (G3) at (4+1,0) {$3$};
\node[fill=green!50, draw=black, circle, minimum size=5mm, inner sep=0.5pt] (G4) at (4+1,1) {$3$};
\node[fill=blue!50, draw=black, circle, minimum size=5mm,  inner sep=0.5pt] (G5) at (4+2,0) {$2$};
\node[fill=red!50, draw=black, circle, minimum size=5mm, inner sep=0.5pt] (G6) at (4+2,1) {$1$};

\draw[very thick] (G1)--(G2)--(G4)--(G6)--(G5)--(G3)--(G1);
\draw[very thick] (G1)--(G4);
\draw[very thick] (G3)--(G6);

\draw[very thick, dotted] (6-0.4,1.2)  to[out=90, in=90] (6+0.4,1.2) -- (6+0.4, -0.2)  to[out=-90, in=-90] (6-0.4, -0.2) -- cycle;


\node[fill=red!50, draw=black, circle, minimum size=5mm, inner sep=0.5pt] (G1) at (8+0,0) {$1$};
\node[fill=green!50, draw=black, circle, minimum size=5mm,  inner sep=0.5pt] (G2) at (8+0,1) {$3$};
\node[fill=green!50, draw=black, circle, minimum size=5mm, inner sep=0.5pt] (G3) at (8+1,0) {$3$};
\node[fill=blue!50, draw=black, circle, minimum size=5mm, inner sep=0.5pt] (G4) at (8+1,1) {$2$};
\node[fill=blue!50, draw=black, circle, minimum size=5mm,  inner sep=0.5pt] (G5) at (8+2,0) {$2$};
\node[fill=red!50, draw=black, circle, minimum size=5mm, inner sep=0.5pt] (G6) at (8+2,1) {$1$};

\draw[very thick] (G1)--(G2)--(G4)--(G6)--(G5)--(G3)--(G1);
\draw[very thick] (G1)--(G4);
\draw[very thick] (G3)--(G6);

\draw[very thick, dotted] (8-0.2,1-0.4)  to[out=180, in=180] (8-0.2,1+0.4) -- (9+0.2, 1+0.4)  to[out=0, in=0] (9+0.2, 1-0.4) -- cycle;


\node[fill=red!50, draw=black, circle, minimum size=5mm, inner sep=0.5pt] (G1) at (12+0,0) {$1$};
\node[fill=green!50, draw=black, circle, minimum size=5mm,  inner sep=0.5pt] (G2) at (12+0,1) {$3$};
\node[fill=blue!50, draw=black, circle, minimum size=5mm, inner sep=0.5pt] (G3) at (12+1,0) {$2$};
\node[fill=blue!50, draw=black, circle, minimum size=5mm, inner sep=0.5pt] (G4) at (12+1,1) {$2$};
\node[fill=green!50, draw=black, circle, minimum size=5mm,  inner sep=0.5pt] (G5) at (12+2,0) {$3$};
\node[fill=red!50, draw=black, circle, minimum size=5mm, inner sep=0.5pt] (G6) at (12+2,1) {$1$};

\draw[very thick] (G1)--(G2)--(G4)--(G6)--(G5)--(G3)--(G1);
\draw[very thick] (G1)--(G4);
\draw[very thick] (G3)--(G6);

\node[] at (13,-0.75) {$f_t$};

\draw[latex-latex, very thick] (2.5,0.5) -- (3.5,0.5);
\draw[latex-latex, very thick] (6.5,0.5) -- (7.5,0.5);
\draw[latex-latex, very thick] (10.5,0.5) -- (11.5,0.5);

\draw[very thick, dotted] (13-0.2,0-0.4)  to[out=180, in=180] (13-0.2,0+0.4) -- (14+0.2, 0+0.4)  to[out=0, in=0] (14+0.2, 0-0.4) -- cycle;
\end{tikzpicture}}
	\caption{A reconfiguration sequence between two proper $3$-colorings $f_s$ and $f_t$ under color swapping.}
	\label{fig:reconf_example}
\end{figure}

\subsection{Our Contribution}
\label{subsec:contributions}
\begin{figure}
	\centering
	\hspace{-17mm} \resizebox{0.9\textwidth}{!}{\usetikzlibrary{arrows.meta, positioning, shapes.multipart, calc}
\begin{tikzpicture}[]

\draw[very thick, blue!80] (-7.5,-3.7)--(5,-3.7);
\draw[very thick, blue!80, dashed] (-7.5,-2.1)--(5,-2.1);
\draw[very thick, blue!80, dotted] (-7.5,-0.6)--(0,-0.6)--(0,-1.7)--(5,-1.7);

\node[text width=1.7cm, align=left] at (-7.3,-0.1) {\textcolor{blue!80}{$\threeCRCS$ \\ $\PSPACE$-c.}};
\node[text width=1.7cm, align=left] at (-7.3,-1.6) {\textcolor{blue!80}{$\kCRCS$ \\ $\PSPACE$-c.}};
\node[text width=1.7cm, align=left] at (-7.3,-3.2) {\textcolor{blue!80}{$\CRCS$ \\ $\PSPACE$-c.}};

\draw[-latex, blue!80, ultra thick] (-6.4,-3.7)--(-6.4, -3);
\draw[-latex, blue!80, ultra thick, dashed] (-6.4,-2.1)--(-6.4, -1.4);
\draw[-latex, blue!80, ultra thick, dotted] (-6.4,-0.6)--(-6.4, 0.1);

\draw[very thick, red!80] (-8.2,-3.8) rectangle (1,-5.5);
\node[text width=1.7cm, align=left] at (-7.3,-4.5) {\textcolor{red!80}{$\CRCS$ \\ Poly. time}};

\draw[very thick, red!80, dotted] (1.6,-3.8) rectangle (5,-5.5);
\node[text width=1.7cm, align=right] at (4,-4.5) {\textcolor{red!80}{$\threeCRCS$ \\ Linear time}};

\draw[very thick, red!80, dashed] (-6.1,-2.2) rectangle (-2.6,-3.6);
\node[text width=1.7cm, align=right] at (-3.5,-2.9) {\textcolor{red!80}{$\kCRCS$ \\ Poly. time \\ \text{[Cor.~\ref{col:kCRCSpolynomial}]}}};

\node[](general) at (0,0) {general};
\node[](perfect) at (-1.5,-1) {perfect};
\node[text width=1.3cm, align=center](chordal) at (-3.5,-1.6) {chordal \\ \textcolor{blue!80}{[Thm.~\ref{thm:kCRCS_PSPACEcomp_chordal}]}};
\node[text width=1.3cm, align=center](split) at (-5.4,-3) {split \\ \textcolor{blue!80}{[Thm.~\ref{thm:CRCS_PSPACEcomp_split}]}};
\node[text width=1.3cm, align=center](bipartite) at (1.25,-3) {bipartite \\ \textcolor{blue!80}{[Thm.~\ref{thm:CRCS_PSPACEcomp_bipartite}]}};
\node[text width=1.3cm, align=center](cograph) at (-2.5,-4.5) {cograph \\ \textcolor{red!80}{[Thm.~\ref{thm:cographs}]}};
\node[](threshold) at (-5.4,-5) {threshold};
\node[text width=1.3cm, align=center](completebipartite) at (-0,-5) {complete \\ bipartite};
\node[text width=4cm, align=center](planar) at (3,-1) {planar~$\cap$~max. degree $3$ \\ \textcolor{blue!80}{[Thm.~\ref{thm:3CRCS_PSPACE-hard}]}};
\node[text width=1.3cm, align=center](path) at (2.3,-4.75) {path \\ \textcolor{red!80}{[Thm.~\ref{thm:path_linear}]}};

\draw[->] (general) -- (perfect);
\draw[->] (perfect)--(chordal);
\draw[->] (chordal)--(split);
\draw[->] (split)--(threshold);
\draw[->] (general) -- (planar);
\draw[->] (perfect) -- (cograph);
\draw[->] (cograph)--(threshold);
\draw[->] (cograph) -- (completebipartite);
\draw[->] (perfect) -- (bipartite);
\draw[->] (bipartite)--(completebipartite);
\draw[->] (bipartite)--(path);
\draw[->] (chordal)--(path);
\draw[->] (planar)--(path);

\end{tikzpicture}}
	\caption{Our results for graph classes. Each arrow represents the inclusion relationship between classes: $A\rightarrow B$ means that the graph class $B$ is a proper subclass of the graph class $A$.}
	\label{fig:results}
\end{figure}

In this paper, we introduce a new reconfiguration rule, called \emph{color swapping}~($\CS$) for {\prb{Coloring Reconfiguration}}.
A new proper $k$-coloring is obtained from a given one by swapping the colors of the endpoints of a single edge $uv$ in $G$, so that the resulting coloring remains proper.
For example, \Cref{fig:reconf_example} shows a reconfiguration sequence between $f_s$ and $f_t$ under {$\CS$}, hence it is a yes-instance.
We refer to \prb{Coloring Reconfiguration} and {\kCR} under $\CS$ as {\CRCS} and {\kCRCS}, respectively.
The color swapping rule can be seen as a restricted variant of Kempe chain recoloring, where each Kempe chain is limited to exactly two vertices.
Interestingly, this rule is also studied in statistical physics as \emph{(local) Kawasaki dynamics}, which models dynamics in the fixed-magnetization Ising model~\cite{CarlsonDKP22,Kawasaki66,Kawasaki:KuchukovaP0Y24}.

The contribution of this paper is an analysis of the computational complexity of {\CRCS} and {\kCRCS}; for an overview of our results, we refer to \Cref{fig:results}.
First, we prove that {\CRCS} is $\PSPACE$-complete even when the input graph is restricted to bipartite or split graphs.
Furthermore, we show that there exists a positive integer $k_0$ such that for every fixed $k \geq k_0$, {\kCRCS} becomes $\PSPACE$-complete even when the input graph is restricted to chordal graphs, which is a superclass of split graphs.

We also establish a complexity dichotomy with respect to the number{~$k$} of colors:
{\kCRCS} is $\PSPACE$-complete for any fixed $k \geq 3$, whereas for $k \leq 2$, the problem can be solved in polynomial time.
In particular, we show that $\threeCRCS$ is $\PSPACE$-complete even when restricted to planar graphs with maximum degree $3$ and bounded bandwidth.

Complementing these hardness results, we also present several positive results.
We first show that $\threeCRCS$ can be solved in linear time on path graphs.
To this end, we introduce an invariant for proper $3$-colorings of paths, and design a linear-time algorithm that checks whether two input colorings have the identical invariants.
While the algorithm is simple, its correctness requires a non-trivial argument.

Next, we show that {\CRCS} can be solved in polynomial time on cographs.
Our algorithm is based on a recursive procedure over the cotree of the input cograph, inspired by prior work~\cite{DdTS:KaminskiMM12} for \prb{Independent Set Reconfiguration} on cographs.
To adapt this approach to our problem, we introduce a new notion called \emph{extended $k$-colorings} as a generalization of $k$-colorings.

Finally, we show that {\kCRCS} on split graphs is polynomial-time solvable for any fixed~$k$.
This contrasts with the $\PSPACE$-hardness of {\CRCS} on split graphs when~$k$ is unbounded.


\subsection{Related Work}
As mentioned in \Cref{subsec:intro_coloring_reconf}, the complexity of {\kCR} has been extensively investigated with respect to various graph classes.
Under both the single-vertex recoloring and Kempe chain recoloring rules, the problem is known to be $\PSPACE$-complete in general.
Moreover, stronger hardness results and polynomial-time algorithms have been established for specific graph classes.

Under the single-vertex recoloring rule, the problem remains $\PSPACE$-complete even on bipartite planar graphs~\cite{single_recolor:BonsmaC09} and chordal graphs~\cite{single_recolor:HatanakaIZ19}.
On the positive side, it is solvable in polynomial time for $2$-degenerate graphs, $q$-trees (for fixed $q$), trivially perfect graphs, and split graphs~\cite{single_recolor:HatanakaIZ19}.
Under the Kempe chain recoloring rule, the problem is $\PSPACE$-complete even on planar graphs with maximum degree~$6$~\cite{Kempe:BonamyHIKMMSW20}; however, it is {polynomial-time solvable} on chordal graphs, bipartite graphs, and cographs~\cite{Kempe:BonamyHIKMMSW20}.
Several other algorithmic and structural aspects of {\kCR} have also been studied, including finding a shortest reconfiguration sequence~\cite{BonsmaMNR14,ColorSwap:JohnsonKKPP16} and bounding its length~\cite{BartierBFHMP23,BonamyB13,BousquetH22,CerecedaHJ08}.

The term \emph{color swapping} is inspired by the \prb{Colored Token Swapping} problem~\cite{color_token:BonnetMR18,color_token:YamanakaHKKOSUU18}, a reconfiguration problem involving token placements.
In that problem, one is given a graph with an initial and a target colorings, which \emph{need not be proper}, and the goal is to transform the initial coloring into the target one using the minimum number of swaps between tokens on adjacent vertices.
Since feasibility is not restricted to proper colorings, this can be viewed as a variant of our reconfiguration problem in which the feasibility condition is relaxed.
Yamanaka et al.~\cite{color_token:YamanakaHKKOSUU18} showed that \prb{Colored Token Swapping} is $\NP$-hard even when the number of colors is exactly three.

\section{Preliminaries}
For a positive integer~$k$, we write $[k] = \{1,2,\ldots,k\}$.
For sets $X$ and $Y$, the \emph{symmetric difference} of $X$ and $Y$ is defined as $X \symdif Y \coloneqq (X \setminus Y) \cup (Y \setminus X)$.
For a map $f \colon X \to Y$ and an element $y \in Y$, the \emph{preimage} of $y$ under~$f$ is defined as $f^{-1}(y) \coloneqq \{x \in X \mid f(x) = y\}$.

Let $G = (V, E)$ be an undirected graph.
We use $V(G)$ and $E(G)$ to denote the vertex set and edge set of~$G$, respectively.
For a vertex $v$ of $G$, $N_G(v)$ and $N_G[v]$ denote the \emph{open neighborhood} and the \emph{closed neighborhood} of $v$ in $G$, respectively; that is, $N_G(v) = \{u \in V \mid uv \in E\}$ and $N_G[v] = N_G(v) \cup \{v\}$.
For a vertex set $X\subseteq V$, we define $N_G(X)=\{v\in V \setminus X \mid  u\in X,uv\in E(G)\}$ and $N_G[X]=N_G(X)\cup X$.

For a positive integer $k$, a \emph{(proper) $k$-coloring} of a graph $G$ is a map $f\colon V(G)\to [k]$ that assigns different colors to adjacent vertices; in other words, $f(u) \neq f(v)$ for every edge $uv \in E(G)$.
An \emph{independent set} of a graph $G$ is a subset $I\subseteq V(G)$ such that for all $u,v\in I$, $uv\notin E(G)$ holds.
A \emph{clique} of a graph $G$ is a subset $C\subseteq V(G)$ such that for all $u,v\in C$ with $u\neq v$, $uv\in E(G)$ holds.

\subsection{Our Problems}\label{subsec:ourprb}
For two proper colorings $f$ and $f'$ of a graph $G$, we say that $f$ and $f'$ are \emph{adjacent} under \emph{Color Swapping} (denoted by~$\CS$) if and only if the following two conditions hold:
\begin{enumerate}
	\item There exists exactly one edge $uv \in E(G)$ such that $f(u) = f'(v)$ and $f(v) = f'(u)$, and
	\item For all other vertices $w \in V(G) \setminus \{u, v\}$, we have $f(w) = f'(w)$.
\end{enumerate}
Intuitively, $f'$ can be obtained from $f$ by swapping the colors of two adjacent vertices.

A sequence $f_0, f_1, \ldots, f_\ell$ of proper $k$-colorings of~$G$ with $f_0 = f$ and $f_\ell = f'$ is called a \emph{reconfiguration sequence} between $f$ and $f'$ if $f_{i-1}$ and $f_i$ are adjacent under~$\CS$ for all $i \in [\ell]$.
We say that $f$ and $f'$ are \emph{reconfigurable} if such a sequence exists.
We now define the \prb{Coloring Reconfiguration under Color Swapping} problem ({\CRCS} for short).
In {\CRCS}, we are given a graph $G$, a positive integer $k$, and two proper $k$-colorings $f_s$ and $f_t$ of $G$.
The goal is to determine whether there exists a reconfiguration sequence between $f_s$ and $f_t$.
For a fixed positive integer~$k$, the {\kCRCS} problem is the special case of {\CRCS} in which the two input colorings are $k$-colorings.

An instance $(G, k, f_s, f_t)$ of {\CRCS} is said to be \emph{valid} if, for each color $i \in [k]$, the number of vertices assigned color~$i$ is the same in $f_s$ and $f_t$; that is, $|f_s^{-1}(i)| = |f_t^{-1}(i)|$.
Note that if $f_s$ and $f_t$ are reconfigurable, then {$(G, k, f_s, f_t)$} must be valid.
This observation implies that if an instance $(G, k, f_s, f_t)$ of {\CRCS} is not valid, we can immediately return \rev{\texttt{NO}}.
Therefore, throughout this paper, we assume that all instances of {\CRCS} are valid.

It is easy to see that {\kCRCS} for $k \leq 2$ can be solved in polynomial time.

\begin{observation}
	\label{obs:2CRCSsolvability}
	{\kCRCS} can be solved in polynomial time for $k \leq 2$.
\end{observation}

\begin{proof}
	First, consider the case $k = 1$.
	Since the input graph is edgeless and all proper $1$-colorings assign the same color to every vertex, the answer is always \rev{\texttt{YES}}.

	Next, consider the case $k = 2$.
	Let $(G, f_s, f_t)$ be an instance of $2$-$\CRCS$.
	Without loss of generality, we assume that $G$ is connected; otherwise, each connected component can be handled independently.
	If $|V(G)| \leq 2$, then any two $2$-colorings are reconfigurable, and hence the \rev{instance} is a yes-instance.
	Now suppose $|V(G)| \geq 3$.
	We observe that \rev{one cannot apply any color-swapping operations.}
	Consequently, $f_s$ and $f_t$ are reconfigurable if and only if $f_s = f_t$.
\end{proof}

\section{PSPACE-completeness}\label{sec:PSPACE-comp.}
We first observe that {\CRCS} is solvable using polynomial space, i.e., {\CRCS} belongs to the class $\PSPACE$.
This follows from the equivalence $\PSPACE = \NPSPACE$, which is a consequence of Savitch’s theorem~\cite{ColorSwap:Savitch70}.
To see the membership of $\PSPACE$, consider a reconfiguration sequence for {\CRCS} as a certificate.
Our polynomial-space algorithm reads each $k$-coloring in the sequence one by one and verifies that (i) each coloring is proper and (ii) each pair of consecutive colorings is adjacent under~$\CS$.
Since each of these checks can be performed using polynomial space, this implies that {\CRCS} can be solved in nondeterministic polynomial space.
Therefore, by $\PSPACE = \NPSPACE$, we conclude that {\CRCS} is in $\PSPACE$.

In this section, we present polynomial-time reductions from three different problems to establish the $\PSPACE$-completeness of our problem for several graph classes.
Specifically,
we reduce from the following problems: the \prb{Token Sliding} problem in \Cref{subsec:ReductionfromTS}, the \prb{Coloring Reconfiguration} problem in \Cref{subsec:ReductionfromCR}, and the \prb{Nondeterministic Constraint Logic} problem in \Cref{subsec:ReductionfromNCL}.

\subsection{Reduction from Token Sliding}\label{subsec:ReductionfromTS}
In this subsection, we present a polynomial-time reduction from the \prb{Token Sliding} problem to \CRCS.
\prb{Token Sliding} is also known as the \prb{Independent Set Reconfiguration under Token Sliding}~\cite{EMPIS:HearnD05,DdTS:KaminskiMM12}.

Let $G$ be a graph, and let $I \subseteq V(G)$ be a vertex subset of $G$.
Recall that $I$ is an \emph{independent set} of $G$ if no two vertices in $I$ are adjacent; that is, $uv \notin E(G)$ for all $u,v \in I$.
Two independent sets $I$ and $J$ of the same size are said to be \emph{adjacent} under \emph{token sliding} if and only if $|I \symdif J| = 2$ and the two vertices in $I \symdif J$ are joined by an edge in $G$.

In the \prb{Token Sliding} problem, we are given a graph $G$ and two independent sets $I_s, I_t \subseteq V(G)$ of the same size.
The goal is to determine whether there exists a sequence $I_0, I_1, \ldots, I_\ell$ of independent sets of $G$ such that $I_s = I_0$ and $I_t = I_\ell$, and for every $i \in [\ell]$, $I_{i-1}$ and $I_i$ are adjacent.

\prb{Token Sliding} is known to be $\PSPACE$-complete even when restricted to split graphs~\cite{IndepSet:BelmonteKLMOS21} or bipartite graphs~\cite{IndepSet:LokshtanovM19}.
A graph is \emph{split} if its vertices can be partitioned into a clique and an independent set, and \emph{bipartite} if its vertices can be partitioned into two independent sets.

We first show the following theorem for split graphs.

\begin{theorem}\label{thm:CRCS_PSPACEcomp_split}
	{\CRCS} is $\PSPACE$-complete for split graphs.
\end{theorem}
\begin{proof}
	We have already observed that the problem is in $\PSPACE$.
	To show the $\PSPACE$-hardness, we give a polynomial-time reduction from \prb{Token Sliding} on split graphs to {\CRCS}.

	Let $(G, I_s, I_t)$ be an instance of \prb{Token Sliding} such that $|I_s| = |I_t|$ and $G = (V, E)$ is a split graph with vertex partition $(C, S)$, where $C$ is a clique of $G$ and $S$ is an independent set of $G$.
	Assume that $|I_s|=|I_t|\geq 2$.
	Note that \prb{Token Sliding} remains $\PSPACE$-complete under this assumption since the problem is trivial when $|I_s|=|I_t|= 1$.

	We construct an instance $(G', k, f_s, f_t)$ of {\CRCS} as follows.
	Let $G'$ be the graph obtained from $G$ by adding a new vertex $u$ that is adjacent to all vertices in $V(G)$.
	That is, let $G' = (V', E')$ with $V' = V \cup \{u\}$ and $E' = E \cup \{uv \mid v \in V\}$.
	Since $u$ is adjacent to all vertices in $C$, the set $C \cup \{u\}$ is a clique of $G'$, and $S$ is an independent set of $G'$.
	Thus, $G'$ is also a split graph.

	Let $k =  |V'| - |I_s| + 1 = |V'| - |I_t| + 1$.
	We now define proper $k$-colorings $f_s$ and $f_t$ of $G'$.
	For each $v \in I_s$, set $f_s(v) = 1$.
	Then, assign a distinct color from $[k] \setminus \{1\}$ arbitrarily to each vertex in $V' \setminus I_s$ so that $f_s$ is a proper $k$-coloring.
	Similarly, define $f_t$ by setting $f_t(v) = 1$ for each $v \in I_t$, and assigning a distinct color from $[k] \setminus \{1\}$ arbitrarily to each vertex in $V' \setminus I_t$ so that $f_t$ is also a proper $k$-coloring.
	Since $|V' \setminus I_s| = k - 1$ and $|V' \setminus I_t| = k - 1$, such proper $k$-colorings $f_s$ and $f_t$ exist.

	This completes the construction of the instance $(G', k, f_s, f_t)$.
	We claim that $(G, I_s, I_t)$ is a yes-instance of \prb{Token Sliding} if and only if $(G', k, f_s, f_t)$ is a yes-instance of {\CRCS}.

	We first show the ``only if'' direction.
	Suppose that $(G, I_s, I_t)$ is a yes-instance of \prb{Token Sliding}.
	Then, there exists a sequence of independent sets $I_0, I_1, \ldots, I_\ell$ of $G$, with $I_0 = I_s$ and $I_\ell = I_t$, and for each $i\in[\ell]$, $I_{i-1}$ and $I_{i}$ are adjacent under {\TS}.
	For each $i\in\{0,1,\ldots,\ell\}$, we construct a proper $k$-coloring $f_i$ of $G'$ from $I_i$, following the same construction as for $f_s$ and $f_t$: for each $v\in I_i$, set $f_i(v)=1$, and assign a distinct color from $[k] \setminus \{1\}$ arbitrarily to each vertex in $V' \setminus I_i$ so that $f_i$ is a proper $k$-coloring {of $G'$}.

	We claim that for all $i\in[\ell]$, $f_{i-1}$ and $f_{i}$ are reconfigurable under $\CS$.
	Let $f'$ be the coloring obtained from $f_{i-1}$ by exchanging the colors assigned to the vertices $v \in I_{i-1} \setminus I_{i}$ and $w \in I_{i} \setminus I_{i-1}$.
	Since each vertex in $I_{i - 1}$ is assigned the color $1$ in $f_{i -1}$, and all other vertices are assigned distinct colors from $[k] \setminus \{1\}$, the modified coloring $f'$ remains a proper $k$-coloring of $G'$.
	Moreover, since $I_{i-1}$ and $I_{i}$ are adjacent under {\TS}, we have $vw \in E(G')$.
	Thus, $f_{i-1}$ and $f'$ are adjacent under $\CS$.

	We now show that $f'$ and $f_{i}$ are reconfigurable under $\CS$.
		{Recall that $f'(v) = f_{i}(v) = 1$ for all $v \in I_{i}$.
			Thus, $f'$ and $f_{i}$ may differ only on the vertices in $V(G') \setminus I_{i}$.}
	Note that the restrictions of $f'$ and $f_i$ to $V(G') \setminus I_i$ are both bijective.

	By construction, $G'[V(G') \setminus I_{i}]$ is connected, since it contains a universal vertex $u$ adjacent to all others.
	It is known that for any connected $n$-vertex graph and two bijective $n$-colorings, there exists a reconfiguration sequence of length $O(n^2)$ between them under $\CS$~\cite{ColorSwap:YamanakaDIKKOSS15}.
	Applying this result, we can reconfigure $f'$ into $f_{i}$ using color swaps only within $V(G') \setminus I_{i}$.
	Thus, $f_{i-1}$ and $f_{i}$ are reconfigurable under $\CS$.
	This implies that we can construct a reconfiguration sequence between $f_s$ and $f_t$ under $\CS$.
	Therefore, $(G', k, f_s, f_t)$ is a yes-instance of $\CRCS$.

	We now prove the ``if'' direction.
	Suppose that there exists a reconfiguration sequence between $f_s$ and $f_t$ under $\CS$.
	Let $f_0,f_1,\ldots,f_\ell$ be the reconfiguration sequence, where $f_0=f_s$ and $f_\ell=f_t$.
	For each $i\in \{0,1,\ldots,\ell\}$, define $I_i=f_i^{-1}(1)$.
	Since the number of vertices colored~$1$ remains constant throughout the sequence, it follows that $|I_i| = |I_s| = |I_t| \geq 2$ for all $i$.
	By construction, the vertex $u$ is adjacent to all other vertices in $G'$, and thus cannot be assigned color $1$ in any $f_i$; hence, $I_i \subseteq V(G)$.
	Moreover, since $f_i$ is a proper $k$-coloring of $G$, the set $I_i$ forms an independent set of $G'$, {and consequently} also an independent set of $G$.

	For each $i\in[\ell]$, since two consecutive proper $k$-colorings $f_{i-1}$ and $f_{i}$ are adjacent under $\CS$, we can observe that $I_{i-1}$ and $I_{i}$ are either identical or adjacent under {\TS}.
	By deleting any consecutive duplicate independent sets from the sequence $I_0,I_1,\ldots,I_\ell$, we obtain the desired sequence between $I_s$ and $I_t$.
	Therefore, $(G, I_s, I_t)$ is a yes-instance of \prb{Token Sliding}.

	This completes the proof of \Cref{thm:CRCS_PSPACEcomp_split}.
\end{proof}

{The similar result holds for bipartite graphs.}

\begin{theorem}\label{thm:CRCS_PSPACEcomp_bipartite}
	{\CRCS} is $\PSPACE$-complete for bipartite graphs.
\end{theorem}

\begin{proof}
	We have already observed that the problem is in $\PSPACE$.
	To show the $\PSPACE$-hardness, we give a polynomial-time reduction from \prb{Token Sliding} on bipartite graphs to {\CRCS}.

	Let $(G, I_s, I_t)$ be an instance of \prb{Token Sliding}, where $G$ is a bipartite graph with a bipartition {$(S, T)$}, and both $S$ and $T$ are independent sets {of $G$}.
	We construct an instance $(G', k, f_s, f_t)$ of $\CRCS$ as follows (see also \Cref{fig:TStoCRCS_bipartite}).

	First, we construct $G'$ by adding three new vertices $x_1, x_2, x_3$ to $S$, and three new vertices $y_1, y_2, y_3$ to $T$.
	We then make $x_1$ adjacent to all vertices in $T \cup \{y_1,y_2, y_3\}$, and $y_1$ adjacent to all vertices in $S \cup \{x_1, x_2, x_3\}$.
	Since $S \cup \{x_1, x_2, x_3\}$ and $T \cup \{y_1, y_2, y_3\}$ are independent sets of $G'$, the resulting graph $G'$ is also bipartite.

	We set $k = |V(G')| - |I_s| - 3$, and define proper $k$-colorings $f_s$ and $f_t$ of $G'$ as follows.
	For the initial coloring $f_s$, assign color $1$ to every vertex in $I_s \cup \{x_2, x_3, y_2, y_3\}$.
	Then, assign a distinct color from $[k] \setminus \{1\}$ arbitrarily to each vertex in $V(G') \setminus (I_s \cup \{x_2, x_3, y_2, y_3\})$ so that $f_s$ is a proper $k$-coloring {of $G'$}.
	Similarly, for the target coloring $f_t$, assign color $1$ to every vertex in $I_t \cup \{x_2, x_3, y_2, y_3\}$.
	Then, assign a distinct color from $[k] \setminus \{1\}$ arbitrarily to each vertex in $V(G') \setminus (I_t \cup \{x_2, x_3, y_2, y_3\})$ such that $f_t$ is also a proper $k$-coloring {of $G'$}.
	Note that both $V(G') \setminus (I_s \cup \{x_2, x_3, y_2, y_3\})$ and $V(G') \setminus (I_t \cup \{x_2, x_3, y_2, y_3\})$ contain exactly $(k - 1)$ vertices.
		{Moreover, both $(I_s \cup \{x_2, x_3, y_2, y_3\})$ and $(I_t \cup \{x_2, x_3, y_2, y_3\})$ form independent sets of $G'$.}
	Thus, it is always possible to assign the remaining $(k - 1)$ colors so that both $f_s$ and $f_t$ are proper $k$-colorings of $G'$.

	This completes the construction {of the instance $(G', k, f_s, f_t)$}.
	We then claim that $(G, I_s, I_t)$ is a yes-instance of \prb{Token Sliding} if and only if $(G', k, f_s, f_t)$ is a yes-instance of $\CRCS$.

	We first show the ``only if'' direction.
	Suppose that $(G, I_s, I_t)$ is a yes-instance of \prb{Token Sliding}.
	Then, there exists a sequence of independent sets $I_0, I_1, \ldots, I_\ell$ of $G$, with $I_0 = I_s$ and $I_\ell = I_t$, and for each $i\in[\ell]$, $I_{i-1}$ and $I_{i}$ are adjacent under {\TS}.
	For each $i\in\{0,1,\ldots,\ell\}$, we construct a proper $k$-coloring $f_i$ of $G'$ from $I_i$, following the same construction as for $f_s$ and $f_t$.
    To be precise, for the coloring $f_i$, assign color $1$ to every vertex in $I_i \cup \{x_2, x_3, y_2, y_3\}$.
	Then, assign a distinct color from $[k] \setminus \{1\}$ arbitrarily to each vertex in $V(G') \setminus (I_i \cup \{x_2, x_3, y_2, y_3\})$ such that $f_i$ is also a proper $k$-coloring.

	We claim that $f_{i -1}$ and $f_{i}$ are reconfigurable under $\CS$.
	Let $f'$ be the coloring obtained from $f_{i-1}$ by exchanging the colors assigned to the vertices $v \in I_{i-1} \setminus I_{i}$ and ${w} \in I_{i} \setminus I_{i-1}$.
	Since each vertex in $I_{i}\cup\{x_2,x_3,y_2,y_3\}$ is assigned the color $1$, and all other vertices are assigned distinct colors from $[k] \setminus \{1\}$ by $f'$, the modified coloring $f'$ remains a proper $k$-coloring of $G'$.
	Moreover, since $I_{i - 1}$ and $I_{i}$ are adjacent under {\TS}, we have $v{w} \in E(G')$.
	Hence, $f_i$ and $f'$ are adjacent under $\CS$.

	\begin{figure}[t]
		\centering
		\scalebox{0.8}{\begin{tikzpicture}[scale=0.85]

\node[] at (0,-0.75) {$S$};
\node[fill=white!50, draw=black, circle, minimum size=6mm] (S1) at (0,0) {};
\node[fill=white!50, draw=black, circle, minimum size=6mm] (S2) at (0,1) {};
\node[fill=black!100, draw=black, circle, minimum size=6mm] (S3) at (0,2) {};
\node[fill=black!100, draw=black, circle, minimum size=6mm] (S4) at (0,3) {};

\node[] at (2,-0.75) {$T$};
\node[fill=black!100, draw=black, circle, minimum size=6mm] (T1) at (2,0) {};
\node[fill=white!50, draw=black, circle, minimum size=6mm] (T2) at (2,1) {};
\node[fill=white!50, draw=black, circle, minimum size=6mm] (T3) at (2,2) {};
\node[fill=white!50, draw=black, circle, minimum size=6mm] (T4) at (2,3) {};

\draw[very thick] (S1)--(T1) (S1)--(T2) (S1)--(T3);
\draw[very thick] (S2)--(T1) (S2)--(T3);
\draw[very thick] (S3)--(T2) (S3)--(T3) (S3)--(T4);
\draw[very thick] (S4)--(T3) (S4)--(T4);

\node[] at (1,-1.5) {(a)};

\node[] at (7,-0.75) {$S$};
\node[fill=white!50, draw=black, circle, minimum size=6mm] (C_S1) at (7,0) {$2$};
\node[fill=white!50, draw=black, circle, minimum size=6mm] (C_S2) at (7,1) {$3$};
\node[fill=red!50, draw=black, circle, minimum size=6mm] (C_S3) at (7,2) {$1$};
\node[fill=red!50, draw=black, circle, minimum size=6mm] (C_S4) at (7,3) {$1$};

\node[] at (9,-0.75) {$T$};
\node[fill=red!50, draw=black, circle, minimum size=6mm] (C_T1) at (9,0) {$1$};
\node[fill=white!50, draw=black, circle, minimum size=6mm] (C_T2) at (9,1) {$4$};
\node[fill=white!50, draw=black, circle, minimum size=6mm] (C_T3) at (9,2) {$5$};
\node[fill=white!50, draw=black, circle, minimum size=6mm] (C_T4) at (9,3) {$6$};

\draw[very thick] (C_S1)--(C_T1) (C_S1)--(C_T2) (C_S1)--(C_T3);
\draw[very thick] (C_S2)--(C_T1) (C_S2)--(C_T3);
\draw[very thick] (C_S3)--(C_T2) (C_S3)--(C_T3) (C_S3)--(C_T4);
\draw[very thick] (C_S4)--(C_T3) (C_S4)--(C_T4);

\node[fill=white!50, draw=black, circle, minimum size=6mm, label=below:$y_1$] (y1) at (7,4.5) {$7$};
\node[fill=red!50, draw=black, circle, minimum size=6mm, label=above:$x_2$] (x2) at (7,5.5) {$1$};
\node[fill=red!50, draw=black, circle, minimum size=6mm, label=above:$x_3$] (x3) at (6,5.5) {$1$};

\node[fill=white!50, draw=black, circle, minimum size=6mm, label=below:$x_1$] (x1) at (9,4.5) {$8$};
\node[fill=red!50, draw=black, circle, minimum size=6mm, label=above:$y_2$] (y2) at (9,5.5) {$1$};
\node[fill=red!50, draw=black, circle, minimum size=6mm, label=above:$y_3$] (y3) at (10,5.5) {$1$};

\draw[very thick] (y1) to[out=200, in=135] (C_S1);
\draw[very thick] (y1) to[out=200, in=135] (C_S2);
\draw[very thick] (y1) to[out=200, in=135] (C_S3);
\draw[very thick] (y1) to[out=200, in=135] (C_S4);

\draw[very thick] (x1) to[out=-20, in=45] (C_T1);
\draw[very thick] (x1) to[out=-20, in=45] (C_T2);
\draw[very thick] (x1) to[out=-20, in=45] (C_T3);
\draw[very thick] (x1) to[out=-20, in=45] (C_T4);

\draw[very thick] (x1)--(y1);
\draw[very thick] (x1)--(y2);
\draw[very thick] (x1)--(y3);
\draw[very thick] (y1)--(x2);
\draw[very thick] (y1)--(x3);

\node[] at (8,-1.5) {(b)};

\end{tikzpicture}}
		\caption{%
			(a) The bipartite graph~$G$ for an instance of the \prb{Token Sliding} problem, along with its initial independent set~$I$, where the vertices in~$I$ are marked in black. The vertex set of the graph~$G$ is partitioned into independent sets $S$ and $T$.
			(b) The bipartite graph~$G'$, constructed from~$G$, has a bipartition $V(G') = (S \cup \{x_1,x_2,x_3\}) \cup (T \cup \{y_1,y_2,y_3\})$, where both parts are independent sets. Each vertex is labeled with its color in the initial proper $k$-coloring~$f_s$ derived from~$I$.
		}
		\label{fig:TStoCRCS_bipartite}
	\end{figure}
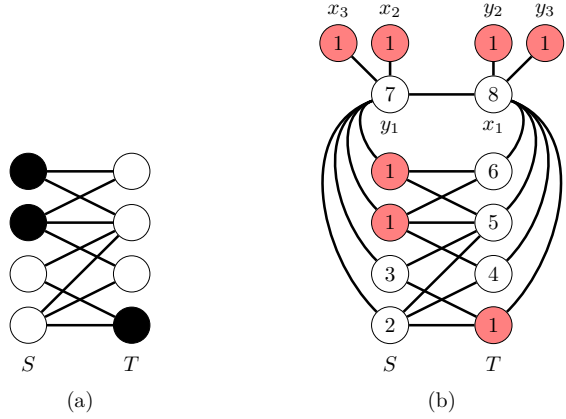

	We now prove that $f'$ and $f_{i}$ are reconfigurable under $\CS$.
		{Recall that $f'(v) = f_{i}(v) = 1$ for all $v \in I_{i}\cup \{x_2,x_3,y_2,y_3\}$}
	Thus, $f'$ and $f_{i}$ may differ only on the vertices in $[k] \setminus \{1\}$.
	Note that the restrictions of $f'$ and $f_i$ to $V(G') \setminus I_i$ are both bijective.

	Recall that $x_1$ is adjacent to every vertex in $T \cup \{y_1, y_2, y_3\}$ and $y_1$ is adjacent to every vertex in $S \cup \{x_1, x_2, x_3\}$, so the subgraph induced by $V(G') \setminus (I_{i} \cup \{x_2, x_3, y_2, y_3\})$ is connected.
	As shown in the proof of \Cref{thm:CRCS_PSPACEcomp_split} and using the result from~\cite{ColorSwap:YamanakaDIKKOSS15}, {we see that} $f'$ and $f_{i}$ are reconfigurable under $\CS$, and thus so are $f_{i-1}$ and $f_{i}$.
	By repeating this process for all $i \in [\ell]$, we obtain a reconfiguration sequence from $f_s$ to $f_t$ under $\CS$.
	Therefore, $(G', k, f_s, f_t)$ is a yes-instance of $\CRCS$.

	We now prove the ``if'' direction.
	Suppose that there exists a reconfiguration sequence between $f_s$ and $f_t$ under $\CS$.
	Let $f_0,f_1,\ldots,f_\ell$ be the reconfiguration sequence, where $f_0=f_s$ and $f_\ell=f_t$.

	Since $x_2$ and $x_3$ share the same neighbor $y_1$ and are both assigned color~$1$ in $f_s$, no color swap involving $x_2$ or $x_3$ is allowed throughout the sequence; that is, $f_i(x_2) = f_i(x_3) = 1$ for all $i \in \{0,1,\ldots,\ell\}$.
		{By symmetry, we also have} $f_i(y_2) = f_i(y_3) = 1$, {which implies that} $x_1$ and $y_1$ are never assigned color~1.
	For each $i \in \{0,1,\ldots,\ell\}$, define $I_i = f_i^{-1}(1)$ and $I'_i = I_i \setminus \{x_2, x_3, y_2, y_3\}$.
	Then, $|I'_i| = |I_s| = |I_t|$, and since $I_i$ is an independent set in $G'$, {it follows that} $I'_i$ is also an independent set in $G$.

	Furthermore, since $f_{i-1}$ and $f_{i}$ are adjacent under $\CS$ for each $i \in [\ell]$, the sets $I'_{i - 1}$ and $I'_{i}$ are either adjacent under {\TS} or identical.
	By removing consecutive duplicates from the sequence $I'_0, I_1,\ldots, I'_\ell$, we obtain a reconfiguration sequence of independent sets from $I_s$ to $I_t$ under {\TS}.
	Therefore, $(G, I_s, I_t)$ is a yes-instance of \prb{Token Sliding}.

	This completes the proof of \Cref{thm:CRCS_PSPACEcomp_split}.
\end{proof}

\subsection{Reduction from Coloring Reconfiguration}\label{subsec:ReductionfromCR}
In this subsection, we {give a polynomial-time reduction from the {\kCR} problem under single-vertex recoloring to {\kCRCS}}.

Let $k$ be a positive integer.
In {\kCR} {under single-vertex recoloring}, we are given a graph $G$ and two $k$-colorings $g$ and $g'$ of $G$.
The goal is to determine whether there exists a sequence of $k$-colorings $g_0, g_1, \ldots, g_\ell$ with $g_0 = g$ and $g_\ell = g'$ such that for each $i \in [\ell]$, the colorings $g_{i-1}$ and $g_i$ differ at exactly one vertex; that is, $|\{v \in V(G) \mid g_{i-1}(v) \ne g_i(v)\}| = 1$.
	{We simply call the problem {\kCR}.}
It is known that there exists a positive integer $k_0$ such that, for any fixed $k \geq k_0$, {\kCR} is $\PSPACE$-complete on chordal graphs~\cite{ECR:HatanakaIZ19}.
Recall that a graph is \emph{chordal} if it contains no induced cycle of length at least~$4$.

We claim the following theorem.
\begin{theorem}
	\label{thm:kCRCS_PSPACEcomp_chordal}
	There exists a positive integer $k_0$ such that, for any fixed $k \geq k_0$, $\kCRCS$ is $\PSPACE$-complete on chordal graphs.
\end{theorem}
\begin{proof}
	We have already observed that the problem is in $\PSPACE$.
	To show the $\PSPACE$-hardness, we give a polynomial-time reduction from {\kCR} {on chordal graphs} to {\kCRCS}.

	Let $(G, g_s, g_t)$ be an instance of {\kCR}, where $G$ is a chordal graph.
	We construct an instance $(G', f_s, f_t)$ of $\kCRCS$ as follows (see also \Cref{fig:CRtoCRCS}).
	For each vertex $v \in V(G)$, we add $(k - 1)$ new vertices $v_1, {v_2}, \ldots, v_{k-1}$ and make $\{v, v_1, {v_2}, \ldots, v_{k-1}\}$ a clique in $G'$.
	Let $C_v = \{v, v_1, {v_2},\ldots, v_{k-1}\}$.
	Note that the resulting graph $G'$ is also chordal.
	We then define the $k$-coloring $f_s$ as follows: for each $v \in V(G)$, set $f_s(v) = g_s(v)$, and for each {vertex in} $C_v \setminus \{v\}$, assign an arbitrary color distinct from $f_s(v)$ so that $f_s$ is a proper $k$-coloring of $G$ (which is always possible since $|C_v| = k$).
	Similarly, define $f_t(v) = g_t(v)$ for each $v \in V(G)$, and for each {vertex in} $C_v \setminus \{v\}$, assign an arbitrary color distinct from $g_t(v)$ so that $f_t$ is a proper $k$-coloring of $G$.

	This completes the construction {of the instance $(G', f_s, f_t)$}.
	We then claim that $(G, g_s, g_t)$ is a yes-instance of {\kCR} if and only if $(G', f_s, f_t)$ is a yes-instance of $\kCRCS$.

	We first prove the ``only if'' direction.
	Suppose that $(G, g_s, g_t)$ is a yes-instance of {\kCR}.
	Then there exists a reconfiguration sequence $g_0, g_1, \ldots, g_\ell$ of proper $k$-colorings of $G$ such that $g_0 = g_s$, $g_\ell = g_t$, and for each $i \in [\ell]$, {$g_{i - 1}$ and $g_{i}$} differ at exactly one vertex.
	For each $i \in \{0, 1, \ldots, \ell\}$, we construct a proper $k$-coloring $f_i$ of $G'$ from $g_i$, following the same construction as for $f_s$ and $f_t$: for every $v \in V(G)$, set $f_i(v) = g_i(v)$, and for each $u \in C_v \setminus \{v\}$, assign an arbitrary color distinct from $f_i(v)$ so that $f_i$ is a proper $k$-coloring of $G'$.

	We claim that for all $i \in [\ell]$, $f_{i-1}$ and $f_i$ are reconfigurable under $\CS$.
		{Indeed, let $u \in V(G)$ be the unique vertex such that $g_{i-1}(u) \neq g_i(u)$.
			By construction, we have $f_{i-1}(u) \neq f_i(u)$, while $f_{i-1}(v) = f_i(v)$ for all $v \in V(G) \setminus \{u\}$.}
	Let $w \in C_u$ be the unique vertex such that $f_{i-1}(w) = f_i(u)$.
	We construct a $k$-coloring $f'$ from $f_{i-1}$ by swapping the colors of $u$ and $w$.
	Note that $f'(v) = f_i(v)$ for all $v \in V(G) \setminus \{u\}$.

	Recall that for each $v \in V(G)$, the clique $C_v$ in $G'$ contains exactly $k$ vertices, and in both colorings $f_i$ and $f'$, {the vertices of $C_v$ receive pairwise distinct colors.}
	Since vertices in $C_v \setminus \{v\}$ are adjacent only to those in $C_v$, we can freely perform color swaps within $C_v$ without affecting outside $C_v$.
	Thus, we can reconfigure $f'$ into $f_i$ by performing a sequence of color swaps only within cliques $C_v$ for $v \in V(G)$.
	This implies that $f_{i-1}$ and $f_i$ are reconfigurable under $\CS$.
		{Applying this argument for each $i \in [\ell]$ yields} a reconfiguration sequence from $f_s$ to $f_t$ in $G'$ under $\CS$.
	Therefore, $(G', f_s, f_t)$ is a yes-instance of $\kCRCS$.

	We now prove the ``if'' direction.
	Suppose that there exists a reconfiguration sequence $f_0, f_1, \ldots, f_\ell$ between $f_s$ and $f_t$, where $f_0 = f_s$ and $f_\ell = f_t$.
	For each $i \in \{0,1,\ldots,\ell\}$, define the coloring $g_i$ of $G$ by setting $g_i(v) = f_{i}(v)$ for every $v \in V(G)$.
	Since $f_{i}$ is a proper $k$-coloring of $G'$, the construction {guarantees} that $g_i$ is a proper $k$-coloring of $G$.

	Recall that $C_w$ for each $w \in V(G)$ is a clique of size $k$.
	Thus, any color swap occurs only on two vertices within some clique $C_w$.
	It follows that for every $i \in [\ell]$, the colorings $g_{i-1}$ and $g_i$ are either identical or differ at exactly one vertex of $G$.
	By removing any consecutive duplicate colorings, we obtain a desired sequence of proper $k$-colorings of $G$ from $g_s$ to $g_t$.
	Therefore, $(G, g_s, g_t)$ is a yes-instance of {\kCR}.

	This completes the proof of \Cref{thm:kCRCS_PSPACEcomp_chordal}.
\end{proof}

\begin{figure}[t]
	\centering
\begin{tikzpicture}[scale=0.8]

\node[fill=red!50, draw=black, circle, minimum size=5mm, inner sep=1pt] (v1) at (-1-4.5,0) {$1$};
\node[fill=blue!50, draw=black, circle, minimum size=5mm, inner sep=1pt] (v2) at (-1-3,1) {$2$};
\node[fill=green!50, draw=black, circle, minimum size=5mm, inner sep=1pt] (v3) at (-1-3,-1) {$3$};
\node[fill=yellow!50, draw=black, circle, minimum size=5mm, inner sep=1pt] (v4) at (-1-1.5,0) {$4$};

\draw[very thick] (v1)--(v2);
\draw[very thick] (v1)--(v3);
\draw[very thick] (v2)--(v3);
\draw[very thick] (v2)--(v4);
\draw[very thick] (v3)--(v4);
\node[] at (-1-3,-2) {$G$};

\node[fill=red!50, draw=black, circle, minimum size=5mm, inner sep=1pt] (u1) at (1.5,0) {$1$};
\node[fill=blue!50, draw=black, circle, minimum size=5mm, inner sep=1pt] (u2) at (3,1) {$2$};
\node[fill=green!50, draw=black, circle, minimum size=5mm, inner sep=1pt] (u3) at (3,-1) {$3$};
\node[fill=yellow!50, draw=black, circle, minimum size=5mm, inner sep=1pt] (u4) at (4.5,0) {$4$};

\draw[very thick] (u1)--(u2);
\draw[very thick] (u1)--(u3);
\draw[very thick] (u2)--(u3);
\draw[very thick] (u2)--(u4);
\draw[very thick] (u3)--(u4);

\node[fill=blue!50, draw=black, circle, minimum size=3mm] (k1) at (0.5,0) {};
\node[fill=green!50, draw=black, circle, minimum size=3mm] (k2) at (0.9,0.5) {};
\node[fill=yellow!50, draw=black, circle, minimum size=3mm] (k3) at (0.9,-0.5) {};

\node[fill=red!50, draw=black, circle, minimum size=3mm] (k4) at (3,2) {};
\node[fill=green!50, draw=black, circle, minimum size=3mm] (k5) at (3.5,1.6) {};
\node[fill=yellow!50, draw=black, circle, minimum size=3mm] (k6) at (2.5,1.6) {};

\node[fill=red!50, draw=black, circle, minimum size=3mm] (k7) at (3,-2) {};
\node[fill=blue!50, draw=black, circle, minimum size=3mm] (k8) at (3.5,-1.6) {};
\node[fill=yellow!50, draw=black, circle, minimum size=3mm] (k9) at (2.5,-1.6) {};

\node[fill=blue!50, draw=black, circle, minimum size=3mm] (k10) at (5.5,0) {};
\node[fill=green!50, draw=black, circle, minimum size=3mm] (k11) at (5.1,0.5) {};
\node[fill=red!50, draw=black, circle, minimum size=3mm] (k12) at (5.1,-0.5) {};

\draw[very thick] (u1)--(k2)--(k1)--(k3)--(u1);
\draw[very thick] (u1)--(k1);
\draw[very thick] (k2)--(k3);

\draw[very thick] (u2)--(k5)--(k4)--(k6)--(u2);
\draw[very thick] (u2)--(k4);
\draw[very thick] (k5)--(k6);

\draw[very thick] (u3)--(k8)--(k7)--(k9)--(u3);
\draw[very thick] (u3)--(k7);
\draw[very thick] (k8)--(k9);

\draw[very thick] (u4)--(k11)--(k10)--(k12)--(u4);
\draw[very thick] (u4)--(k10);
\draw[very thick] (k11)--(k12);

\node[] at (3,-2.75) {$G'$};

\end{tikzpicture}
	\caption{{
				Construction of $G$ from $G'$ using four colors.
				Vertices of $G$ are assigned a proper $4$-coloring $g_t$, and those of $G'$ the corresponding proper $4$-coloring $f_t$.}}
	\label{fig:CRtoCRCS}
\end{figure}
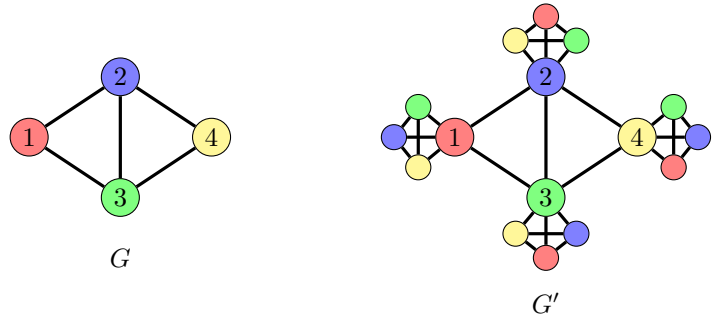

\subsection{Reduction from Nondeterministic Constraint Logic}\label{subsec:ReductionfromNCL}
In this subsection, we prove \Cref{thm:3CRCS_PSPACE-hard} by a polynomial-time reduction from the \prb{Nondeterministic Constraint Logic} problem{~\cite{HearnDemaine_book,Zanden15}}.

\begin{theorem}\label{thm:3CRCS_PSPACE-hard}
	For every fixed integer $k \geq 3$, $\kCRCS$ is $\PSPACE$-complete for planar graphs of maximum degree $3$ and bounded bandwidth.
\end{theorem}

We begin by formally defining the {\prbNCL} problem in \Cref{subsubsec:NCL_explanation}.
Next, we introduce an auxiliary gadget used in our reduction, which is described in \Cref{subsubsec:forbidden_structure}.
We then present the full reduction in \Cref{subsubsec:NCL_Gadgets_reduction}, including the design of two types of gadgets: \textsc{and} gadgets, and \textsc{or} gadgets.
Finally, we prove the correctness of the reduction in \Cref{subsubsec:NCL_correctness}.

Note that our reduction is presented for $k = 3$; however, it holds analogously for cases where $k \geq 4$.

\subsubsection{Definition of Nondeterministic Constraint Logic}\label{subsubsec:NCL_explanation}

A \emph{Nondeterministic Constraint Logic (NCL) machine} is an undirected graph in which each edge is assigned a weight from $\{1, 2\}$ (see \Cref{fig:NCL}~(a)).
A \emph{configuration} of an NCL machine is an orientation of its edges such that, at every vertex, the total weight of incoming edges is at least $2$.
Two configurations are said to be \emph{adjacent} if they differ in the orientation of exactly one edge.
In the \prb{Nondeterministic Constraint Logic} problem, we are given an NCL machine $M$ and two of its configurations, $C_s$ and $C_t$.
The goal is to determine whether there exists a sequence of configurations starting from $C_s$ and ending at $C_t$, such that each consecutive pair of configurations in the sequence differs in the orientation of exactly one edge; that is, they are adjacent.

An NCL machine $M$ is called an \emph{\textsc{and/or} constraint graph} if it contains only two types of vertices: \emph{\textsc{and} vertices} and \emph{\textsc{or} vertices} (see again \Cref{fig:NCL}~(a)).
A vertex of degree~$3$ in $M$ is an \emph{\textsc{and} vertex} if its three incident edges have weights $1$, $1$, and $2$; see \Cref{fig:NCL}~(b).
Similarly, a vertex of degree~$3$ in $M$ is defined as an \emph{\textsc{or} vertex} if all of its three incident edges have weight $2$; see \Cref{fig:NCL}~(c).
For the remainder of this paper, we will use the term \emph{NCL machine} to refer specifically to an \textsc{and/or} constraint graph.
It is known that the {\prbNCL} problem remains $\PSPACE$-complete for NCL machines {(\textsc{and/or} constraint graphs)} are restricted to be planar, of maximum degree $3$, and of bounded bandwidth{~\cite{Zanden15}}.
Recall that, for a graph $G$, the \emph{bandwidth} of $G$ is the minimum integer $b$ such that $G$ has a bijection $\pi \colon V(G)\to [|V(G)|]$ with $\max_{uv\in E(G)}|\pi(u)-\pi(v)| \leq b$.

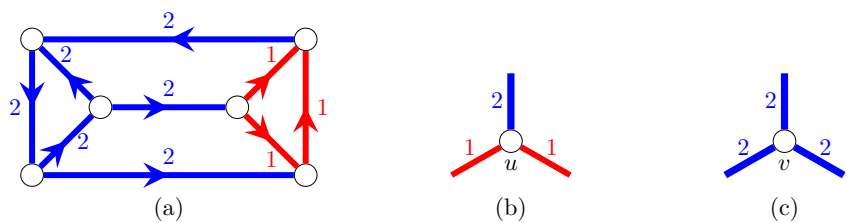
\begin{figure}[t]
	\centering
	\scalebox{0.9}{
\begin{tikzpicture}[scale=1]
\node[fill=white!10, draw=black, circle, minimum size=3mm] (N1) at (-9,1.5){};
\node[fill=white!10, draw=black, circle, minimum size=3mm] (N2) at (-9,-0.5){};
\node[fill=white!10, draw=black, circle, minimum size=3mm] (N3) at (-8,0.5){};
\node[fill=white!10, draw=black, circle, minimum size=3mm] (N4) at (-5,1.5){};
\node[fill=white!10, draw=black, circle, minimum size=3mm] (N5) at (-5,-0.5){};
\node[fill=white!10, draw=black, circle, minimum size=3mm] (N6) at (-6,0.5){};

\draw[line width=2.5pt, blue, >=stealth, postaction={decorate}, decoration={markings,mark=at position 0.5 with {\arrow[scale=1.3]{>}}}] (N1)  -- (N2) node[midway, left] {2};
\draw[line width=2.5pt, blue, >=stealth, postaction={decorate}, decoration={markings,mark=at position 0.5 with {\arrow[scale=1.3]{>}}}] (N3)  -- (N1) node[midway, above] {2};
\draw[line width=2.5pt, blue, >=stealth, postaction={decorate}, decoration={markings,mark=at position 0.5 with {\arrow[scale=1.3]{>}}}] (N4)  -- (N1) node[midway, above] {2};
\draw[line width=2.5pt, blue, >=stealth, postaction={decorate}, decoration={markings,mark=at position 0.5 with {\arrow[scale=1.3]{>}}}] (N2)  -- (N3) node[midway, right] {2};
\draw[line width=2.5pt, blue, >=stealth, postaction={decorate}, decoration={markings,mark=at position 0.5 with {\arrow[scale=1.3]{>}}}] (N2)  -- (N5) node[midway, above] {2};
\draw[line width=2.5pt, blue, >=stealth, postaction={decorate}, decoration={markings,mark=at position 0.5 with {\arrow[scale=1.3]{>}}}] (N3)  -- (N6) node[midway, above] {2};

\draw[line width=2.5pt, red, >=stealth, postaction={decorate}, decoration={markings,mark=at position 0.5 with {\arrow[scale=1.3]{>}}}] (N5) -- (N4) node[midway, right] {1};
\draw[line width=2.5pt, red, >=stealth, postaction={decorate}, decoration={markings,mark=at position 0.5 with {\arrow[scale=1.3]{>}}}] (N6) -- (N4)  node[midway, above] {1};
\draw[line width=2.5pt, red, >=stealth, postaction={decorate}, decoration={markings,mark=at position 0.5 with {\arrow[scale=1.3]{>}}}] (N6) -- (N5)  node[midway, below] {1};

\node[] at (-7,-1) {(a)};

\node[fill=white!10, draw=black, circle, minimum size=3mm] (or) at (-2,0){};
\node[] at (-2,-0.35) {$u$};

\draw[line width=3pt, blue] (or) node at (-2.2,0.6) {2} --(-2, 1);
\draw[line width=2.5pt, red] (or) node at (-1.4,-0.1) {1} --(-2+0.866, -0.5);
\draw[line width=2.5pt, red] (or) node at (-2.6,-0.1) {1} --(-2-0.866, -0.5);

\node[] at (-2,-1) {(b)};

\node[fill=white!10, draw=black, circle, minimum size=3mm] (and) at (2,0){};
\node[] at (2,-0.35) {$v$};

\draw[line width=3pt, blue] (and) node at (1.8,0.6) {2} --(2, 1);
\draw[line width=3pt, blue] (and) node at (1.4,-0.1) {2} --(2+0.866, -0.5);
\draw[line width=3pt, blue] (and) node at (2.6,-0.1) {2} --(2-0.866, -0.5);

\node[] at (2,-1) {(c)};

\end{tikzpicture}}
	\caption{(a) A configuration of an NCL machine, (b) an \textsc{and} vertex, and (c) an \textsc{or} vertex.
		Edges of weight~$2$ are shown in blue lines, and edges of weight~$1$ in red lines.
		The NCL machine in (a) is an \textsc{and/or} constraint graph.
	}
	\label{fig:NCL}
\end{figure}

\subsubsection{Auxiliary Gadgets} \label{subsubsec:forbidden_structure}
\begin{figure}[t]
	\centering
	\scalebox{0.9}{\begin{tikzpicture}[scale=1]
\node[fill=white!50, draw=black, circle, minimum size=5mm] (u) at (-1.5,0) {};
\node[] at (-1.5,-0.5) {$x$};

\node[fill=green!50, draw=black, circle, minimum size=5mm] (f1_1) at (0,0) {$3$};
\node[] at (0,-0.5) {$y$};
\node[fill=green!50, draw=black, circle, minimum size=5mm] (f1_3) at (2,0) {$3$};
\node[fill=red!50, draw=black, circle, minimum size=5mm] (f2) at (1,1) {$1$};
\node[fill=blue!50, draw=black, circle, minimum size=5mm] (f3) at (1,-1) {$2$};

\draw[black, very thick] (u)--(f1_1);
\draw[dashed, very thick] (u)--(-2,1);
\draw[dashed, very thick] (u)--(-2,-1);

\draw[black, very thick] (f2)--(f1_1);
\draw[black, very thick] (f2)--(f1_3);

\draw[black, very thick] (f3)--(f1_1);
\draw[black, very thick] (f3)--(f1_3);

\node[] at (0,-1.5) {(a)};

\node[fill=white!50, draw=black, circle, minimum size=5mm] (u') at (5.5,0) {};
\node[] at (5.5,-0.5) {$x$};
\draw[black, very thick] (u')--(6.6,0);
\draw[dashed, very thick] (u')--(5,1);
\draw[dashed, very thick] (u')--(5,-1);

\filldraw[green!50, draw=black] (7,0.4) -- (6.6,0) -- (7,-0.4) -- (7.4,0) -- cycle;
\node at (7,0) {$3$};

\draw[black, ->, very thick] (3,0)--(4,0);
\node[] at (5.75,-1.5) {(b)};

\end{tikzpicture}}
	\caption{
		(a) An illustration of a 3-forbidden pendant for a vertex~$x$, which ensures that~$x$ is never assigned color~$3$. (b) A simplified depiction of the gadget used to represent this pendant.}
	\label{fig:ForbidGad}
\end{figure}
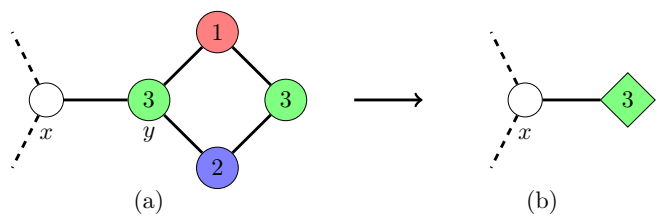

Before presenting our full construction, we introduce auxiliary gadgets, called \emph{forbidden pendants}, which prevent a vertex from being assigned a specific color.

Let $C = \{1,2,3\}$ be a set of colors.
Consider a vertex $x$ adjacent to a vertex $y$, where $y$ is part of a 4-cycle, as illustrated in \Cref{fig:ForbidGad}.
We consider two possible colorings of the 4-cycle in clockwise order starting from $y$: $(1,2,1,3)$ or $(3,1,3,2)$.
Note that in each of these colorings, $y$ is assigned color $1$ or $3$, respectively.

Since $x$ is adjacent to $y$, it must be assigned a color different from that of $y$.
In particular, $x$ can be assigned only colors from the sets $\{2,3\}$ or $\{1,2\}$, respectively, depending on the color of $y$.
Furthermore, observe that no valid color swaps can occur within the cycle or between $x$ and $y$ in any reconfiguration sequence.
This ensures that the color of $y$ remains fixed, and thus $x$ is prevented from ever taking the same color as $y$.

If $y$ is assigned color $c \in \{1, 3\}$, we refer to this gadget as a \emph{$c$-forbidden pendant} for $x$.
For simplicity, we use the diagram shown in \Cref{fig:ForbidGad}~(b) to represent such a gadget.

\subsubsection{AND/OR Gadgets and Our Reduction}\label{subsubsec:NCL_Gadgets_reduction}

\begin{figure}[t]
	\centering
	\includegraphics[scale=1]{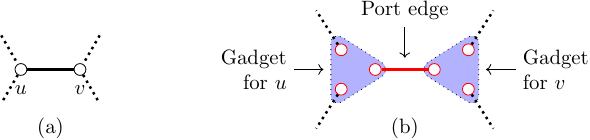}
	\caption{(a) An edge $uv$ in $G$, and (b) its corresponding gadgets, where the port vertices are depicted by red circles and the shared port edge is depicted by a red line.}
	\label{fig:NCL_connect}
\end{figure}

Let $I = (M, C_s, C_t)$ be an instance of {\prbNCL}.
Our reduction constructs a graph $G$ by using two types of \emph{vertex gadgets}, which simulate \textsc{and} and \textsc{or} vertices of $M$, respectively.
	{Each vertex of $M$ is replaced by the corresponding gadget according to its type.}

For each edge $e = uv$ in $M$, the vertex gadgets for $u$ and $v$ each contain a special vertex, called a \emph{port vertex}, which serves as an interface to the corresponding edge.
These two port vertices, denoted by $\portv{u}{e}$ and $\portv{v}{e}$, are connected by an edge referred to as a \emph{port edge} (see \Cref{fig:NCL_connect}).
Accordingly, the gadget corresponding to a vertex $u$ of $M$ with incident edges $e_1$, $e_2$, and $e_3$ includes three port vertices: $\portv{u}{e_1}$, $\portv{u}{e_2}$, and $\portv{u}{e_3}$.

Let $u$ be an \textsc{and} vertex in $M$, incident to one weight-2 edge $e_1$ and two weight-1 edges $e_2$ and $e_3$.
We construct the corresponding \textsc{and} gadget $G_u$, which consists of {two} internal vertices: $u_0$ and $u_2^1$, along with three port vertices: $u_1^1 = \portv{u}{e_1}$, $u_1^2 = \portv{u}{e_2}$, and $u_1^3 = \portv{u}{e_3}$.
Each port vertex of $G_u$ is adjacent to a $3$-forbidden pendant (see \Cref{fig:NCLtoCRCS}~(a)); hence, it can only be colored with color $1$ or $2$ in any proper coloring.
Consequently, any color swap involving a port vertex in $G_u$ must occur along the corresponding port edge.

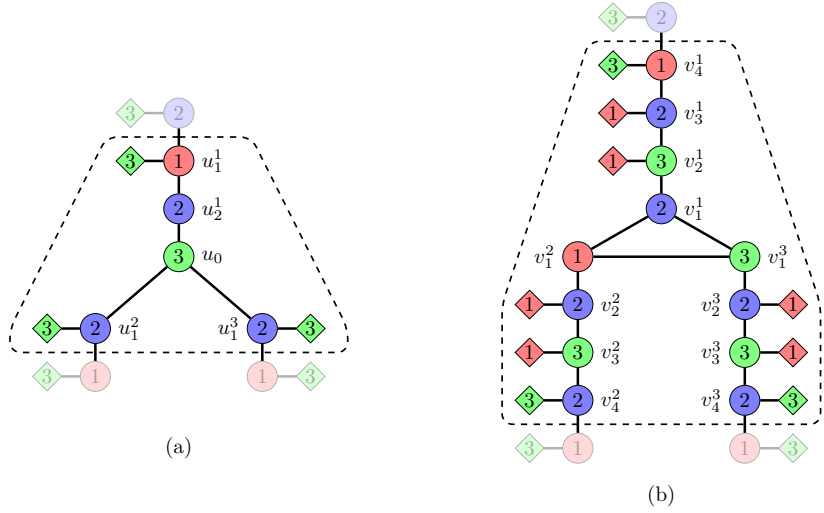
\begin{figure}[t]
	\centering
	\resizebox{0.77\textwidth}{!}{\begin{tikzpicture}[scale=0.8]
\node[fill=blue!50, draw=black, circle, minimum size=5mm, opacity=0.3, inner sep=1pt] (And1) at (-2-3,3) {$2$};
\filldraw[green!50, draw=black, opacity=0.3] (-2-4-0.3,3) -- (-2-4,3-0.3) -- (-2-4+0.3,3) -- (-2-4,3+0.3) -- cycle;
\node[opacity=0.3] at (-2-4,3) {$3$};
\draw[very thick, opacity=0.3] (And1)--(-2-4+0.3,3);

\node[fill=red!50, draw=black, circle, minimum size=5mm,label=right:$u_1^1$, inner sep=1pt] (And2) at (-2-3,2) {$1$};
\filldraw[green!50, draw=black] (-2-4-0.3,2) -- (-2-4,2-0.3) -- (-2-4+0.3,2) -- (-2-4,2+0.3) -- cycle;
\node at (-2-4,2) {$3$};
\draw[very thick] (And2)--(-2-4+0.3,2);

\node[fill=blue!50, draw=black, circle, minimum size=5mm,label=right:$u_2^1$, inner sep=1pt] (And3) at (-2-3,1) {$2$};
\node[fill=green!50, draw=black, circle, minimum size=5mm,label=right:$u_0$, inner sep=1pt] (And4) at (-2-3,0) {$3$};

\node[fill=blue!50, draw=black, circle, minimum size=5mm,label=right:$u_1^2$, inner sep=1pt] (And5) at (-2-3-1.73,-0.5-1) {$2$};
\filldraw[green!50, draw=black] (-2-4-1.73-0.3,-0.5-1) -- (-2-4-1.73,-0.5-0.3-1) -- (-2-4-1.73+0.3,-0.5-1) -- (-2-4-1.73,-0.5+0.3-1) -- cycle;
\node at (-2-4-1.73,-0.5-1) {$3$};
\draw[very thick] (And5)--(-2-4-1.73+0.3,-0.5-1);

\node[fill=red!50, draw=black, circle, minimum size=5mm, opacity=0.3, inner sep=1pt] (And6) at (-2-3-1.73,-1.5-1) {$1$};
\filldraw[green!50, draw=black, opacity=0.3] (-2-4-1.73-0.3,-1.5-1) -- (-2-4-1.73,-1.5-0.3-1) -- (-2-4-1.73+0.3,-1.5-1) -- (-2-4-1.73,-1.5+0.3-1) -- cycle;
\node[opacity=0.3] at (-2-4-1.73,-1.5-1) {$3$};
\draw[very thick, opacity=0.3] (And6)--(-2-4-1.73+0.3,-1.5-1);

\node[fill=blue!50, draw=black, circle, minimum size=5mm,label=left:$u_1^3$, inner sep=1pt] (And7) at (-2-3+1.73,-0.5-1) {$2$};
\filldraw[green!50, draw=black] (-2-2+1.73-0.3,-0.5-1) -- (-2-2+1.73,-0.5-0.3-1) -- (-2-2+1.73+0.3,-0.5-1) -- (-2-2+1.73,-0.5+0.3-1) -- cycle;
\node at (-2-2+1.73,-0.5-1) {$3$};
\draw[very thick] (And7)--(-2-2+1.73-0.3,-0.5-1);

\node[fill=red!50, draw=black, circle, minimum size=5mm, opacity=0.3, inner sep=1pt] (And8) at (-2-3+1.73,-1.5-1) {$1$};
\filldraw[green!50, draw=black, opacity=0.3] (-2-2+1.73-0.3,-1.5-1) -- (-2-2+1.73,-0.5-1.3-1) -- (-2-2+1.73+0.3,-1.5-1) -- (-2-2+1.73,-1.5+0.3-1) -- cycle;
\node[opacity=0.3] at (-2-2+1.73,-1.5-1) {$3$};
\draw[very thick, opacity=0.3] (And8)--(-2-2+1.73-0.3,-1.5-1);

\draw[very thick] (And1)--(And2);
\draw[very thick] (And2)--(And3);
\draw[very thick] (And3)--(And4);
\draw[very thick] (And4)--(And5);
\draw[very thick] (And5)--(And6);
\draw[very thick] (And4)--(And7);
\draw[very thick] (And7)--(And8);

\draw[dashed, thick, rounded corners=5pt] (-2-3,2.5)--(-2-4.5,2.5)--(-2-6.5,-0.5-1)--(-2-6.5,-1-1)--(-2+0.5,-1-1)--(-2+0.5,-0.5-1)--(-2-1.5,2.5)--(-2-3,2.5);

\node[] at (-2-3,-3-1) {(a)};


\node[fill=blue!50, draw=black, circle, minimum size=5mm, opacity=0.3, inner sep=1pt] (Or1) at (5,5) {$2$};
\filldraw[green!50, draw=black, opacity=0.3] (4-0.3,5) -- (4,5-0.3) -- (4+0.3,5) -- (4,5+0.3) -- cycle;
\node[opacity=0.3] at (4,5) {$3$};
\draw[very thick, opacity=0.3] (Or1)--(4+0.3,5);

\node[fill=red!50, draw=black, circle, minimum size=5mm,label=right:$v_4^1$, inner sep=1pt] (Or2) at (5,4) {$1$};
\filldraw[green!50, draw=black] (4-0.3,4) -- (4,4-0.3) -- (4+0.3,4) -- (4,4+0.3) -- cycle;
\node at (4,4) {$3$};
\draw[very thick] (Or2)--(4+0.3,4);

\node[fill=blue!50, draw=black, circle, minimum size=5mm,label=right:$v_3^1$, inner sep=1pt] (Or3) at (5,3) {$2$};
\filldraw[red!50, draw=black] (4-0.3,3) -- (4,3-0.3) -- (4+0.3,3) -- (4,3+0.3) -- cycle;
\node at (4,3) {$1$};
\draw[very thick] (Or3)--(4+0.3,3);

\node[fill=green!50, draw=black, circle, minimum size=5mm,label=right:$v_2^1$, inner sep=1pt] (Or4) at (5,2) {$3$};
\filldraw[red!50, draw=black] (4-0.3,2) -- (4,2-0.3) -- (4+0.3,2) -- (4,2+0.3) -- cycle;
\node at (4,2) {$1$};
\draw[very thick] (Or4)--(4+0.3,2);

\node[fill=blue!50, draw=black, circle, minimum size=5mm,label=right:$v_1^1$, inner sep=1pt] (Or5) at (5,1) {$2$};

\node[fill=red!50, draw=black, circle, minimum size=5mm,label=left:$v_1^2$, inner sep=1pt] (Or6) at (5-1.73,0) {$1$};

\node[fill=blue!50, draw=black, circle, minimum size=5mm,label=right:$v_2^2$, inner sep=1pt] (Or7) at (5-1.73,-1) {$2$};
\filldraw[red!50, draw=black] (4-1.73-0.3,-1) -- (4-1.73,-1-0.3) -- (4-1.73+0.3,-1) -- (4-1.73,-1+0.3) -- cycle;
\node at (4-1.73,-1) {$1$};
\draw[very thick] (Or7)--(4-1.73+0.3,-1);

\node[fill=green!50, draw=black, circle, minimum size=5mm,label=right:$v_3^2$, inner sep=1pt] (Or8) at (5-1.73,-2) {$3$};
\filldraw[red!50, draw=black] (4-1.73-0.3,-2) -- (4-1.73,-2-0.3) -- (4-1.73+0.3,-2) -- (4-1.73,-2+0.3) -- cycle;
\node at (4-1.73,-2) {$1$};
\draw[very thick] (Or8)--(4-1.73+0.3,-2);

\node[fill=blue!50, draw=black, circle, minimum size=5mm,label=right:$v_4^2$, inner sep=1pt] (Or9) at (5-1.73,-3) {$2$};
\filldraw[green!50, draw=black] (4-1.73-0.3,-3) -- (4-1.73,-3-0.3) -- (4-1.73+0.3,-3) -- (4-1.73,-3+0.3) -- cycle;
\node at (4-1.73,-3) {$3$};
\draw[very thick] (Or9)--(4-1.73+0.3,-3);

\node[fill=red!50, draw=black, circle, minimum size=5mm, opacity=0.3, inner sep=1pt] (Or10) at (5-1.73,-4) {$1$};
\filldraw[green!50, draw=black, opacity=0.3] (4-1.73-0.3,-4) -- (4-1.73,-4-0.3) -- (4-1.73+0.3,-4) -- (4-1.73,-4+0.3) -- cycle;
\node[opacity=0.3] at (4-1.73,-4) {$3$};
\draw[very thick, opacity=0.3] (Or10)--(4-1.73+0.3,-4);

\node[fill=green!50, draw=black, circle, minimum size=5mm,label=right:$v_1^3$, inner sep=1pt] (Or11) at (5+1.73,0) {$3$};

\node[fill=blue!50, draw=black, circle, minimum size=5mm,label=left:$v_2^3$, inner sep=1pt] (Or12) at (5+1.73,-1) {$2$};
\filldraw[red!50, draw=black] (6+1.73-0.3,-1) -- (6+1.73,-1-0.3) -- (6+1.73+0.3,-1) -- (6+1.73,-1+0.3) -- cycle;
\node at (6+1.73,-1) {$1$};
\draw[very thick] (Or12)--(6+1.73-0.3,-1);

\node[fill=green!50, draw=black, circle, minimum size=5mm,label=left:$v_3^3$, inner sep=1pt] (Or13) at (5+1.73,-2) {$3$};
\filldraw[red!50, draw=black] (6+1.73-0.3,-2) -- (6+1.73,-2-0.3) -- (6+1.73+0.3,-2) -- (6+1.73,-2+0.3) -- cycle;
\node at (6+1.73,-2) {$1$};
\draw[very thick] (Or13)--(6+1.73-0.3,-2);

\node[fill=blue!50, draw=black, circle, minimum size=5mm,label=left:$v_4^3$, inner sep=1pt] (Or14) at (5+1.73,-3) {$2$};
\filldraw[green!50, draw=black] (6+1.73-0.3,-3) -- (6+1.73,-3-0.3) -- (6+1.73+0.3,-3) -- (6+1.73,-3+0.3) -- cycle;
\node at (6+1.73,-3) {$3$};
\draw[very thick] (Or14)--(6+1.73-0.3,-3);

\node[fill=red!50, draw=black, circle, minimum size=5mm, opacity=0.3, inner sep=1pt] (Or15) at (5+1.73,-4) {$1$};
\filldraw[green!50, draw=black, opacity=0.3] (6+1.73-0.3,-4) -- (6+1.73,-4-0.3) -- (6+1.73+0.3,-4) -- (6+1.73,-4+0.3) -- cycle;
\node[opacity=0.3] at (6+1.73,-4) {$3$};
\draw[very thick, opacity=0.3] (Or15)--(6+1.73-0.3,-4);

\draw[very thick] (Or1)--(Or2)--(Or3)--(Or4)--(Or5)--(Or6)--(Or7)--(Or8)--(Or9)--(Or10);
\draw[very thick] (Or5)--(Or11)--(Or12)--(Or13)--(Or14)--(Or15);
\draw[very thick] (Or6)--(Or11);

\draw[dashed, thick, rounded corners=5pt] (5,4.5)--(3.5,4.5)--(1.7,-1)--(1.7,-3.5)--(8.3,-3.5)--(8.3,-1)--(6.5,4.5)--(5,4.5);

\node[] at (5,-5) {(b)};

\end{tikzpicture}}
	\caption{
		(a) The {\textsc{and}} gadget and (b) the {\textsc{or}} gadget, corresponding to (b) and (c) in \Cref{fig:NCL}, respectively.}
	\label{fig:NCLtoCRCS}
\end{figure}

Let $u$ be an \textsc{or} vertex in $M$, incident to three weight-2 edges $e_1$, $e_2$, and $e_3$.
The corresponding \textsc{or} gadget $G_v$ consists of $12$ vertices, denoted by $v_j^i$ for $i \in [3]$ and $j \in [4]$.
For each $i \in [3]$, the vertex $v_4^i$ is a port vertex of $G_v$, that is, $v_4^i = \portv{v}{e_i}$.
Moreover, every $v_4^i$ is adjacent to a $3$-forbidden pendant, while each of $v_2^i$ and $v_3^i$ is adjacent to a $1$-forbidden pendant (see \Cref{fig:NCLtoCRCS}(b)).

Similar to the \textsc{and} gadget, each port vertex in $G_v$ is restricted to colors $1$ or $2$ due to its adjacency to a $3$-forbidden pendant.
Consequently, any color swap involving a port vertex in $G_v$ must occur along the corresponding port edge.
Moreover, since {both} $v_2^i$ and $v_3^i$ for $i \in [3]$ are adjacent to $1$-forbidden pendants, they must be assigned distinct colors: one must be colored $2$ and the other $3$.
Thus, any color swap involving $v_2^i$ or $v_3^i$ can only occur on the edge $v_2^i v_3^i$.

\paragraph*{Reduction}
Let $I = (M, C_s, C_t)$ be an instance of {\prbNCL}.
We construct a corresponding instance $(G, f_s, f_t)$ of $\kCRCS$.

We begin by constructing a graph $G$ from the NCL instance $M$ as follows.
For each \textsc{and} or \textsc{or} vertex in $M$, we replace it with the corresponding gadget as defined in \Cref{subsubsec:NCL_Gadgets_reduction}.
Then, for each edge $e = uv$ in $M$, we add a port edge between the corresponding port vertices $\portv{u}{e}$ and $\portv{v}{e}$ in the gadgets for $u$ and $v$, respectively; see \Cref{fig:NCL_connect}.

Let $G$ denote the resulting graph.
Then, we observe the following.
\begin{observation}
	\label{obs:resulting_deg_planar_bandwidth}
	The constructed graph $G$ is planar, of maximum degree $3$, and has bounded bandwidth.
\end{observation}

\begin{proof}
	Each $c$-forbidden pendant (for $c \in \{1, 3\}$) has maximum degree $3$, and all gadgets
    \rev{of $G$ also have} maximum degree $3$.
    Thus, the graph $G$ has maximum degree $3$.

	Furthermore, since $M$ is planar and each gadget is planar, we can obtain a planar embedding of $G$.
	Therefore, $G$ is planar.

	Finally, since $M$ has bounded bandwidth and each gadget has constant size, replacing each vertex in $M$ with a gadget increases the bandwidth \rev{by at most a constant factor}.
	Consequently, $G$ has bounded bandwidth.
\end{proof}

We define two proper $3$-colorings, $f_s$ and $f_t$, of the constructed graph $G$.
We begin by assigning colors to all $c$-forbidden pendants for $c \in \{1, 3\}$ according to the colorings described in \Cref{subsubsec:forbidden_structure}.

Let $e$ be an edge of $M$ with an endpoint $u$.
For each corresponding port vertex $\portv{u}{e}$, we set $f_s(\portv{u}{e}) = 1$ (resp.\ $f_t(\portv{u}{e}) = 1$) if the edge $e$ is directed toward $u$ in the configuration $C_s$ (resp.\ $C_t$); otherwise, we assign $f_s(\portv{u}{e}) = 2$ (resp.\ $f_t(\portv{u}{e}) = 2$).

For each \textsc{and} gadget $G_u$ corresponding to an \textsc{and} vertex $u$ of $M$, we assign colors to the internal vertices $u_2^1$ and $u_0$ of $G_u$, which are depicted in \Cref{fig:NCLtoCRCS}, so that $f_s$ becomes a proper $3$-coloring.
That is, we set either $f_s(u_2^1) = 2$ and $f_s(u_0) = 3$, or $f_s(u_2^1) = 3$ and $f_s(u_0) = 2$.
We define $f_t$ similarly to $f_s$ \rev{(see also \Cref{fig:NCLAndConfigurations})}.

For each \textsc{or} gadget $G_v$ corresponding to an \textsc{or} vertex $v$ of $M$, we assign colors to the internal vertices of $G_v$ depending on the coloring of its port vertices under $f_s$ (resp.\ $f_t$).
Specifically, for a given coloring of the port vertices, the internal vertices are colored according to one of the configurations illustrated in \Cref{fig:NCLOrConfigurations}, so that $f_s$ (resp.\ $f_t$) becomes a proper $3$-coloring.
Since multiple valid internal colorings may exist for the same coloring of the port vertices of $G_v$, we may choose any such coloring arbitrarily.

This completes our polynomial-time reduction.

\subsubsection{Correctness}
\label{subsubsec:NCL_correctness}
Before proceeding to our proof, we provide an overview of the main ideas behind our reduction and outline the argument for its correctness.

Our reduction establishes a correspondence between the orientations of edges in a given instance of {\prbNCL} and the colorings in the constructed graph $G$.
For each edge $e = uv$, we have associated two port vertices $\portv{u}{e}$ and $\portv{v}{e}$ with the gadgets for $u$ and $v$, respectively.
We interpret the edge as directed toward vertex $v$ if $\portv{u}{e}$ is assigned color $2$, and as directed toward vertex $u$ if $\portv{v}{e}$ is assigned color $2$.
Note that the coloring must assign color $2$ to exactly one vertex in $\{\portv{u}{e}, \portv{v}{e}\}$ and color $1$ to the other.

We briefly describe the behavior of the \textsc{and} and \textsc{or} gadgets.
\rev{In the \textsc{and} gadget, if the port vertex $u_1^1$ is colored~$1$, then both $u_1^2$ and $u_1^3$ can freely switch their colors.
Conversely, if $u_1^1$ is colored~$2$, then both $u_1^2$ and $u_1^3$ must be colored~$1$.
This behavior mirrors that of an \textsc{and} vertex, where the weight-2 edge (corresponding to $u_1^1$) can be directed outward only if both weight-1 edges (corresponding to $u_1^2$ and $u_1^3$) are directed inward.
See also all proper colorings of the \textsc{and} gadget shown in \Cref{fig:NCLAndConfigurations} to understand the behavior of \textsc{and} gadgets.}

Next, we explain the behavior of the \textsc{or} gadgets.
In the \textsc{or} gadget, it suffices that at least one of the three incident edges is directed inward.
Accordingly, our gadget must only prohibit the configuration in which all three port vertices $v_4^1, v_4^2, v_4^3$ are simultaneously colored with $2$.
As shown in \Cref{fig:NCLtoCRCS}(b), our construction enforces this constraint: if all three port vertices are colored with $2$, then all intermediate vertices $v_2^1, v_2^2, v_2^3$ must also be colored with $2$, which is impossible because the three vertices $v_1^1, v_1^2, v_1^3$ form a clique.
Thus, at least one of the port vertices must be assigned color $1$.
See also all proper colorings of the \textsc{or} gadget shown in \Cref{fig:NCLOrConfigurations} to understand the behavior of \textsc{or} gadgets.

Formally, to establish the correctness of our reduction, we present the following lemma.

\begin{lemma}\label{lem:NCL_correctness}
	The instance $(M, C_s, C_t)$ of {\prbNCL} is a yes-instance if and only if the constructed instance $(G, f_s, f_t)$ of $3$-{\CRCS} is a yes-instance.
\end{lemma}

\begin{proof}
\begin{figure}[t]
    \centering
    \scalebox{0.85}{\input{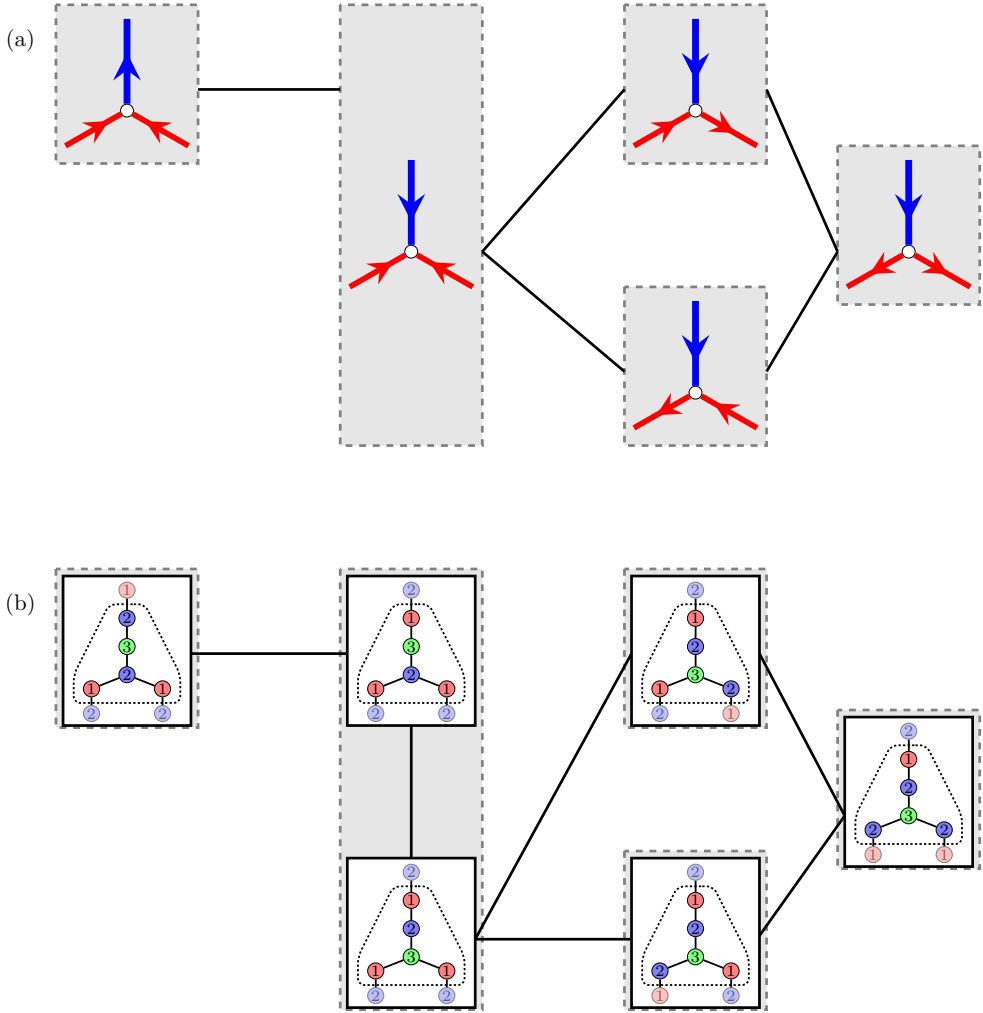}}
    \caption{
    (a) All valid orientations of the three edges incident to an \textsc{and} vertex, together with their adjacency. (b) All corresponding colorings of the \textsc{and} gadget, including the three incident port edges. (The 3-forbidden pendants are omitted for clarity.) 
    \rev{The colorings grouped in the same dotted box correspond to the same configuration of $M$.} 
    }
    \label{fig:NCLAndConfigurations}
\end{figure}

\begin{figure}[t]
	\centering
	\scalebox{0.93}{\input{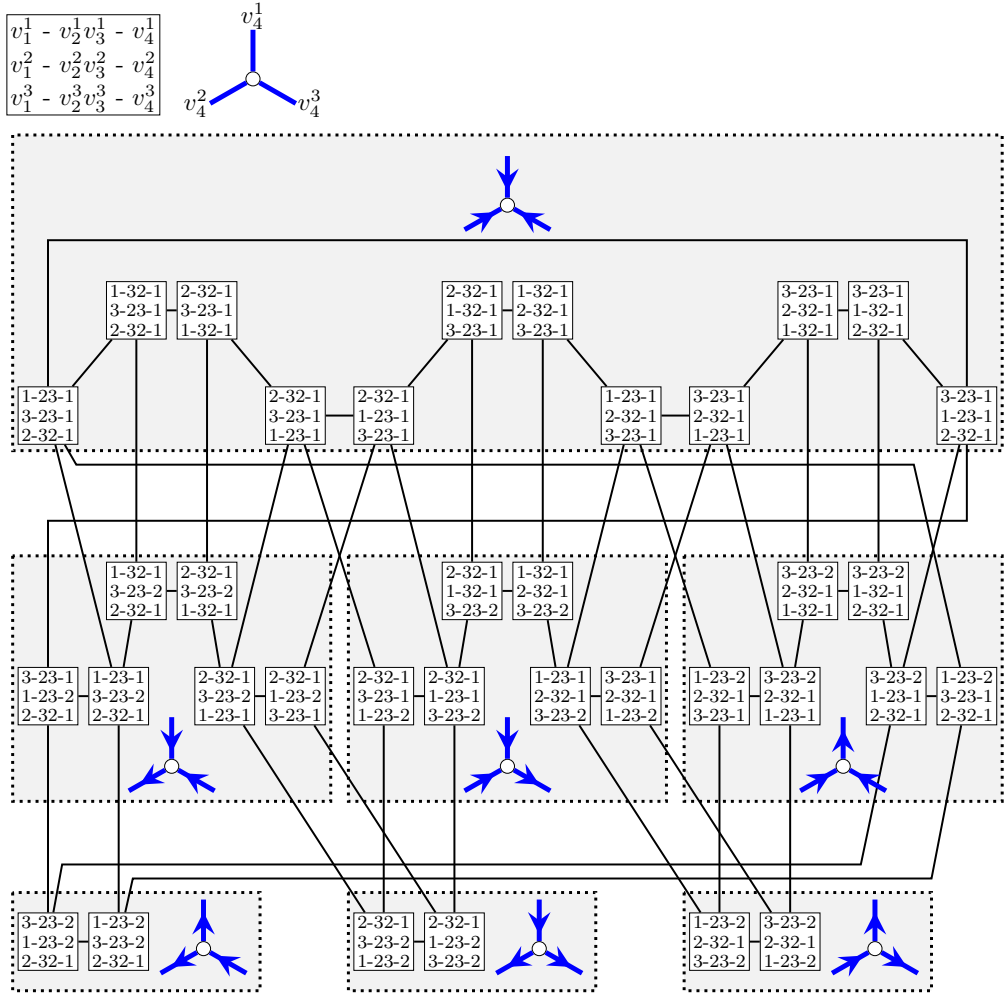}}
	\caption{
		All valid orientations of the three edges incident to an \textsc{or} vertex~$v$, and the corresponding colorings of the \textsc{or} gadget with their adjacency.
		Port vertices $v_4^1$, $v_4^2$, and $v_4^3$ are shown, and only the colors excluding those forbidden by $c$-forbidden pendants are indicated.
		In each row, no color swap is allowed between vertices separated by a hyphen due to the $c$-forbidden pendants.
	}
	\label{fig:NCLOrConfigurations}
\end{figure}

We first prove the ``only if'' direction.  
Suppose that there exists a sequence $C_0, C_1, \ldots, C_{\ell}$ of configurations from $C_s = C_0$ to $C_t = C_\ell$ such that each pair of consecutive configurations in the sequence are adjacent, that is, they differ in the direction of exactly one edge.  
We will show that each \rev{flip a single edge} in the sequence can be simulated by a reconfiguration of proper $k$-colorings of $G$ under the $\CS$.

\rev{For $i = 0, 1, \ldots, \ell$}, let $f_i$ be \rev{a} proper $k$-coloring of $G$ corresponding to $C_i$, which are constructed in the same way as $f_s$ and $f_t$.  
We claim that for all $i \in [\ell]$, $f_{i - 1}$ and $f_{i}$ are reconfigurable under $\CS$.
Suppose $C_{i}$ is obtained from $C_{i - 1}$ by flipping a single edge $e = uv$.
By construction, $f_{i - 1}$ and $f_{i}$ assign the same color to every vertex in each \rev{$1$-forbidden and $3$-forbidden} pendants and to every port vertex, except for $\portv{u}{e}$ and $\portv{v}{e}$.
However, the other vertices in $V(G)$ may differ in color between $f_{i-1}$ and $f_{i}$.
\rev{We claim that the colors of those vertices can be changed to match $f_i$.}

We first consider the \textsc{and} gadgets of $G$.
\Cref{fig:NCLAndConfigurations}~(a) illustrates all valid orientations of the three incident edges for an \textsc{and} vertex, along with the adjacency between these configurations.  
\Cref{fig:NCLAndConfigurations}~(b) shows all \rev{proper} colorings \rev{of} the corresponding \textsc{and} gadget for $u'$, where all $3$-forbidden pendants are omitted for clarity.
There are two corresponding \rev{proper colorings} only when all edges incident to $u'$ are directed inward toward $u'$, and we can \rev{reconfigure} one from the other by swapping the colors of $u_2^1$ and $u_0$ at once.
Moreover, this \rev{reconfiguration can be performed independently of any other gadgets, without requiring additional color swaps elsewhere}.

\rev{The analysis for the \textsc{or} gadget is more involved because of the asymmetry among the colors assigned to $v_1^{1}$, $v_2^{1}$, and $v_3^{1}$}.  
\Cref{fig:NCLOrConfigurations} illustrates all \rev{configurations} of an \textsc{or} vertex $v'$, along with all proper colorings of the corresponding gadget, where all $c$-forbidden pendants with each $c\in \{1,3\}$ are omitted.
In each matrix shown in the figure, the rightmost column represents the colors of the port vertices $v_4^1, v_4^2, v_4^3$, while the leftmost column represents the colors of the triangle vertices $v_1^1, v_1^2, v_1^3$ in the gadget (see also \Cref{fig:NCLOrConfigurations}).  
Colorings enclosed within the same dotted box correspond to the same configuration of $M$.  
Although multiple colorings \rev{of $G$} may correspond to the same configuration \rev{of $M$}, these colorings are reconfigurable under $\CS$ by modifying only the vertices in the gadget.

We now show that for any $i \in [\ell]$, $f_{i - 1}$ and $f_{i}$ are reconfigurable.
\rev{According to \Cref{fig:NCLAndConfigurations,fig:NCLOrConfigurations}}, by performing color swapping on each gadget without involving any port vertex, we can obtain the coloring $f'$ from $f_{i - 1}$ such that $f'(\portv{u}{e})=f_{i}(\portv{v}{e})$, $f'(\portv{v}{e})=f_{i}(\portv{u}{e})$, and $f'(w)=f_{i}(w)$ for all $w\in V(G)\setminus \{\portv{u}{e},\portv{v}{e}\}$. 
Since $f'$ and $f_{i}$ are adjacent under $\CS$, \rev{$f_{i-1}$ and $f_{i}$} are reconfigurable.
Therefore, we \rev{can} construct the reconfiguration sequence between $f_s$ and $f_t$ under $\CS$ from a sequence of configurations between $C_s$ and $C_t$.

We now prove the ``if'' direction.
Suppose that there exists a reconfiguration sequence between $f_s$ and $f_t$ under $\CS$.
Let $f_0, f_1, \ldots, f_\ell$ be such a sequence of proper $3$-colorings of $G$, where $f_0 = f_s$ and $f_\ell = f_t$.
Observe that, in the sequence, color swaps can never involve any vertex that belongs to a $c$-forbidden pendant for $c \in \{1,3\}$. 
Moreover, since every port vertex is adjacent to a $3$-forbidden pendant, each port vertex must be colored either $1$ or $2$ in a coloring $f_i$ for every $i \in [\ell]$.
Consequently, every coloring $f_i$ corresponds to a configuration of the NCL machine \rev{(see again \Cref{fig:NCLAndConfigurations,fig:NCLOrConfigurations})}.

Let $C_0, C_1, \ldots, C_\ell$ denote the sequence of configurations corresponding to $f_0, f_1, \ldots, f_\ell$, respectively.
Since each step in $f_0, f_1, \ldots, f_\ell$ modifies the coloring \rev{by} a single color swap, for $i \in [\ell]$, each configuration $C_i$ differs from $C_{i-1}$ by \rev{the direction of exactly one edge}, or remains the same.
By removing consecutive duplicates from the sequence $C_0, C_1, \ldots, C_\ell$, we obtain a reconfiguration sequence from $C_s$ to $C_t$.

This completes the proof of \Cref{lem:NCL_correctness}.
\end{proof}

\section{Polynomial-time Algorithms}\label{sec:PolyTimeAlgo}
In this section, we present polynomial-time algorithms for {\CRCS} and $\kCRCS$ on paths, cographs, and split graphs.

\subsection{Paths}
We begin by presenting the following result for path graphs:
\begin{theorem}\label{thm:path_linear}
	$\threeCRCS$ can be solved in linear time for path graphs.
\end{theorem}
To prove \Cref{thm:path_linear}, we design a {linear-time} algorithm that solves $\threeCRCS$ on path graphs.

In our algorithm, we compute and compare the invariants of the two input $3$-colorings.
If the invariants are identical, the algorithm returns \texttt{YES}; otherwise, it returns \texttt{NO}.
Although the implementation of our algorithm is relatively simple, we emphasize that the core idea behind the algorithm is conceptually nontrivial, as it captures the essential structure preserved under $\CS$.

\subsubsection{Invariant for Coloring Strings}\label{subsubsec:InvandAlgo}
Before introducing the {\invariant}, we first define several terms used throughout this subsection.

Given a string $S$, $S[i]$ denotes the $i$-th character of $S$, and $S[i,j]$ denotes the substring from the $i$-th to the $j$-th character (inclusive).
If $i > j$, we define $S[i,j] \coloneqq \mathsf{NIL}$, where $\mathsf{NIL}$ denotes the empty string (a string of length $0$).

Let $P = v_1,v_2,\ldots, v_n$ be a path on $n$ vertices.
A string $S = S[1] S[2] \cdots S[n]$ is called a \emph{coloring string} if there exists a proper $3$-coloring $f$ of $P$ such that $S[i] = f(v_i)$ for all $i \in [n]$.
We sometimes refer to $S$ as the coloring string corresponding to $f$.
Note that a proper $3$-coloring $f$ can be encoded into its corresponding coloring string in $O(n)$ time.

We say that two consecutive characters in $S$ are \emph{swappable} if they can be exchanged to produce another coloring string $S'$.
In this case, we say that $S$ and $S'$ are \emph{adjacent}.
If no pair of swappable characters exists in $S$, then we say that $S$ is \emph{rigid}.
Note that each swap of two characters precisely corresponds to a single color swapping operation.
Finally, two coloring strings $S$ and $S'$ are said to be \emph{reconfigurable} if there exists a sequence $S_0, S_1, \ldots, S_\ell$ of coloring strings such that $S_0 = S$, $S_\ell = S'$, and each consecutive pair $S_{i -1}, S_{i}$ is adjacent for all $i \in [\ell]$.

We now define the notion of \emph{contraction}, which will be used to define our invariant for a coloring string.
Let $S = s_1s_2\cdots s_n$ be a coloring string of length $n \ge 3$.
We define the following three contraction operations:
\begin{description}
	\item[(C1.)] If $S = S[1, i-2] \, s_{i-1}s_is_{i+1}s_{i+2} \, S[i+3, n]$ and $s_{i-1} = s_{i+2}$ for some $2 \le i \le n-2$,
	      then $S$ can be contracted to $S[1, i-2] \, s_{i+2} \, S[i+3, n]$.
	\item[(C2.)] If $s_1, s_2, s_3$ are pairwise distinct, then $S = s_1s_2s_3 \, S[4,n]$ can be contracted to $S[4,n]$.
	\item[(C3.)] If $s_{n-2}, s_{n-1}, s_n$ are pairwise distinct, then $S = S[1,n-3] \, s_{n-2}s_{n-1}s_n$ can be contracted to $S[1,n-3]$.
\end{description}
Note that each of these cases removes three pairwise-distinct characters from $S$.
Moreover, any contraction operation on a coloring string $S$ also produces another coloring string of length exactly $|S|-3$.

Let $\mathsf{cont}(S)$ denote the set of all coloring strings that can be obtained from $S$ by applying a single contraction operation.
We now define the invariant $\inv(S)$ of a coloring string $S$ recursively as follows:
\begin{equation*}
	\inv(S) =
	\begin{cases}
		\mathsf{NIL} & \text{if } |S| \le 2,                            \\
		\inv(S')     & \text{if there exists } S' \in \mathsf{cont}(S), \\
		S            & \text{otherwise. \hfill (i.e., $S$ is rigid)}
	\end{cases}
\end{equation*}

The following lemma shows that $\inv(S)$ is uniquely determined regardless of the order in which the contractions are applied.
\begin{lemma} \label{S_well-defined}
	For any coloring string $S$, the invariant $\inv(S)$ is well-defined.
\end{lemma}

\begin{proof}
We prove the claim by induction on the length of $S$.
If $|S| \leq 2$, then by definition \rev{$\inv(S) = \mathsf{NIL}$, and hence $\inv(S)$ is well-defined}.

Assume as the induction hypothesis that the statement holds for all coloring strings of length at most $n - 1$.  
We show that it also holds for strings of length $n$.  
Let $S = s_1s_2\dots s_n$ be a coloring string.

\rev{If $\mathsf{cont}(S) = \emptyset$, then no contraction is applicable to $S$; hence, we have $\inv(S) = S$, which is trivially well-defined.}   
Now suppose that $\mathsf{cont}(S) \neq \emptyset$.  
Let $S'_1, S'_2 \in \mathsf{cont}(S)$ be two \rev{coloring} strings obtained by applying a single contraction to distinct substrings in $S$.  
To prove that $\inv(S)$ is well-defined, it suffices to show that $\inv(S'_1) = \inv(S'_2)$.

If $S'_1 = S'_2$, the claim follows immediately.  
Otherwise, we prove $\inv(S'_1) = \inv(S'_2)$ \rev{by case analysis, reducing each case to the induction hypothesis}. 
\rev{By the symmetry of contraction types~C2 and~C3, we may assume without loss of generality that one of the following holds:
\begin{enumerate}[(a)]
  \item both $S'_1$ and $S'_2$ are obtained by contractions of type C1;
  \item $S'_1$ and $S'_2$ are obtained by contractions of types C1 and C2, respectively; or
  \item $S'_1$ and $S'_2$ are obtained by contractions of types C2 and C3, respectively.
\end{enumerate}
}

{\bf Case (a)}:
Without loss of generality, suppose that $S'_1$ and $S'_2$ are obtained by applying a contraction to the substrings $s_{i-1}s_is_{i+1}$ and $s_{j-1}s_js_{j+1}$, respectively, for some indices $2 \leq i < j \leq n-2$.
Then, we have $S'_1=S[1,i-2]s_{i+2}S[i+3,n]$ and $S'_2=S[1,j-2]s_{j+2}S[j+3,n]$.

When $j=i+1$, by the definition of contraction operations, we have $s_{i-1}=s_{i+2}$.
\rev{Thus},
\begin{align*}
	S'_2 & \rev{= S[1,j-2]s_{j+2}S[j+3,n]} \\
	     & \rev{=S[1,i-1]s_{i+3}S[i+4,n]} \\
	     & = S[1,i-2]s_{i-1}S[i+3,n]  \\
	     & = S[1,i-2]s_{i+2}S[i+3,n] = S'_1.
\end{align*}

When $j=i+2$, by the definition of contraction operations again, we have $s_{i-1} = s_{i+2}$ \rev{and $s_{i+1} = s_{i+4}$, which also implies $s_i = s_{i+3}$}. 
Hence, 
\begin{align*}
    S'_2 & =\rev{S[1,j-2]s_{j+2}S[j+3,n]} \\
         & =\rev{S[1,i]s_{i+4}S[i+5,n]} \\
         & = S[1,i-2]s_{i-1}s_{i}S[i+4,n] \\
         & =\rev{S[1,i-2]s_{i+2}s_{i+3}S[i+4,n]} \\
         & = \rev{S[1,i-2]s_{i+2}S[i+3,n] 
          = S'_1.}
\end{align*}

When $j \ge i + 3$, we can respectively apply a contraction operation to $s_{j-1}s_js_{j+1}$ in $S'_1$ and to $s_{i-1}s_is_{i+1}$ in $S'_2$, 
yielding the same string $S'' = S[1,i-2]S[i+2,j-2]S[j+2,n]$.
Since $S''$ is obtained from both $S'_1$ and $S'_2$ by a single contraction, and $|S'_1| = |S'_2| = n-3$, the induction hypothesis implies $\inv(S'_1)=\inv(S'')=\inv(S'_2)$, which is the desired equality.

{\bf Case (b)}: 
Suppose that $S'_1$ and $S'_2$ are obtained from $S$ by applying a contraction to the substrings $s_1s_2s_3$ and $s_{i-1}s_is_{i+1}$, respectively, for some $2 \leq i \leq n-2$.
When $i = 2$, $S'_1 = S'_2$ clearly holds.

When $i = 3$, we then have $S'_1 = S[4,n]$ and $S'_2 = s_1 S[5,n]$. 
Since both substrings $s_1s_2s_3$ and $s_2s_3s_4$ are removed by contraction operations and all characters in each of these substrings are pairwise distinct, the definition of contraction operations implies $s_1 = s_4$. 
Hence, $S'_1 = s_4S[5, n] = \rev{s_1}S[\rev{5}, n] = S'_2$

When $i=4$, we have $S'_1=s_4s_5s_6S[7,n]$ and $S'_2=s_1s_2s_6S[7,n]$.
Since $s_3 = s_6$ holds by the definition of contraction operations, and both \rev{$(s_4, s_5, s_6)$} and \rev{$(s_1, s_2, s_6)$} consist of pairwise distinct characters, \rev{respectively}. 
Thus, each of $S'_1$ and $S'_2$ can be further removed \rev{by} contraction~C1. 
These contractions yield the same coloring string $S[7,n]$.
Since $|S'_1| = |S'_2| = n -3$, the induction hypothesis implies $\inv(S'_1) = \inv(S[7, n]) = \inv(S'_2)$.

When $i\geq 5$, we have $S'_1=S[4,n]$ and $S'_2=S[1,i-2]S[i+2,n]$.
This situation is analogous to Case~(a) with $j \geq i + 3$, where two non-overlapping contractions result in strings that can be further contracted to a common string.  
By the same reasoning, we conclude $\inv(S'_1)=\inv(S'_2)$.

{\bf Case (c)}: 
Suppose that $S'_1$ and $S'_2$ are obtained from $S$ by applying a contraction operation to the substrings $s_1s_2s_3$ and $s_{n-2}s_{n-1}s_n$, respectively.

When $n\leq 5$, we have $\inv(S)=\mathsf{NIL}$. 
Thus, the invariant is well-defined.

When $n \geq 6$, we have $S'_1=S[4,n]$ and $S'_2=S[1,n-3]$.
By a similar argument in Case~(a) where $j \geq i + 3$, we also observe again $\inv(S'_1)=\inv(S'_2)$.

Therefore, $\inv(S)$ is well-defined.
\end{proof}

A \rev{naive} implementation of computing the invariant of a coloring string would take $O(n^2)$ time, as each contraction step may require scanning the entire string, and up to $O(n)$ such steps may be needed.
However, \rev{the invariant can be computed} in linear time \rev{by simulating the recursive contractions using a stack}.

\begin{lemma}\label{inv_linear_time}
	For any coloring string $S$ of length $n$, the invariant $\inv(S)$ can be computed in $O(n)$ time.
\end{lemma}

\begin{proof}
We now describe how to compute $\inv(S)$ for a given coloring string $S$ of length $n$ in $O(n)$ time using a stack $D$.
\rev{Let $|D|$ denote the number of elements contained in $D$.}

We initialize an empty stack $D$. 
Then, for each index $i = 1$ to $n$, we perform the following operations:
\begin{itemize}
    \item Push $S[i]$ onto the stack $D$.
    \item If $|D| = 3$ and the top three characters \rev{of} $D$ are pairwise distinct, pop all three \rev{of them} (contraction~C2).
    \item If $|D| \ge 4$ and the fourth character from the top \rev{of} $D$ equals $S[i]$, pop the top three characters (contraction~C1).
\end{itemize}
After \rev{scanning} all characters of $S$, we \rev{apply} a post-processing step:
while $|D| \ge 3$ and the top three characters are pairwise distinct, we pop the top three characters \rev{of them} (contraction~C3).
Once the process terminates, if $|D| \le 2$, we return $\mathsf{NIL}$. 
Otherwise, we return the string obtained by concatenating the characters in $D$ from bottom to top.

\rev{Since the algorithm greedily performs each contraction as soon as it becomes applicable, no contraction operation remains applicable to the final stack content.}  
Hence, the \rev{resulting string coincides exactly with $\inv(S)$}.

Each iteration examines at most the top four characters of $D$ and performs either a single push or up to three pops.
\rev{Consequently, the overall running time is $O(n)$.}
\end{proof}

\subsubsection{Correctness of Algorithm}
In the following, we discuss the correctness of our algorithm presented in \Cref{subsubsec:InvandAlgo}.

We begin by observing whether two adjacent characters in a coloring string are swappable.
Recall that a coloring string represents a proper $3$-coloring of an $n$-vertex path.
\begin{observation}\label{obs:swappable}
	Let $S=s_1s_2\cdots s_n$ be a coloring string.
	We can swap $s_i$ and $s_{i+1}$ in $S$ for $i \in [n - 1]$ if and only if the following conditions are satisfied:
	\begin{itemize}
		\item If $i = 1$, then the three characters $s_1,s_2,s_3$ are pairwise distinct.
		\item If $2 \le i \le n - 2$, then $s_{i-1} = s_{i+2}$ holds.
		\item If $i = n - 1$, then the three characters $s_{n-2}, s_{n-1}, s_n$ are pairwise distinct.
	\end{itemize}
\end{observation}

We next present the following lemma, which states that taking a single swap of characters preserves the invariant.

\begin{lemma}\label{lem:same_iq_CS}
	Let $S$ and $S'$ be two adjacent coloring strings.
	Then, $\inv(S) = \inv(S')$.
\end{lemma}

\begin{proof}
If $|S| = |S'| \le 2$, then \rev{$\inv(S) = \inv(S') = \mathsf{NIL}$ by definition, and the claim follows trivially.}

Now consider the case where $n = |S| = |S'| \ge 3$, and suppose that $S'$ is obtained from $S$ by swapping \rev{two adjacent characters $S[i]$ and $S[i+1]$ for some} $1 \le i \le n - 1$.  
When $i = 1$, \rev{\Cref{obs:swappable} implies that} both triples $(S[1], S[2], S[3])$ and $(S'[1], S'[2], S'[3])$ consist of three pairwise distinct characters.
\rev{Hence, contraction~C2 can be applied to $S$ and $S'$, respectively, reducing both strings to the identical suffix $S[4,n]$.}  
It follows that $\inv(S) = \inv(S[4,n]) = \inv(S')$.

The case $i = n - 1$ is symmetric and follows by the same argument.
If $2 \le i \le n - 2$, then \rev{by \Cref{obs:swappable}, we have $S[i-1] = S[i+2] = S'[i-1] = S'[i+2]$}.  
\rev{Therefore, contraction~C1 can be applied to the substrings $S[i-1,i+2]$ and $S'[i-1,i+2]$ in $S$ and $S'$, respectively, after which both strings become identical.}  
Hence, $\inv(S) = \inv(S')$.
\end{proof}

\Cref{lem:same_iq_CS} ensures that if two coloring strings are reconfigurable, then they have the identical invariant.
The next two lemmas are crucial for our converse direction: if two coloring strings have the identical invariant, then they are reconfigurable.

\begin{lemma}\label{lem:collect_top}
	Let $S$ be a coloring string of length at least $3$.
	If $S$ is not rigid, then there exists a coloring string $S'$ such that $S'[1], S'[2], S'[3]$ are pairwise distinct, and $S$ and $S'$ are reconfigurable.
\end{lemma}

\begin{proof}
Let $S$ \rev{be a coloring string of length} $n \ge 3$ such that $S$ is not rigid, and let $i$ be the smallest index for which \rev{$S[i]$ and $S[i+1]$} are swappable.  
\rev{Since $S$ is not rigid, such an index $i$ necessarily exists.} 
We prove the claim by induction on $i$.
If $i \le 2$, then the first three characters, \rev{$S[1]$, $S[2]$, and $S[3]$}, are pairwise distinct, and hence the claim follows immediately.

Now assume that the claim holds for all indices less than some $i \geq 3$. 
Since \rev{$S[i]$ and $S[{i+1}]$} are swappable, it must hold that \rev{$S[{i-1}] \ne S[{i+1}]$}.  
We now show that \rev{$S[{i-2}] = S[i]$} and suppose for contradiction that \rev{$S[{i-2}] \ne S[i]$}.
Then both triples $(S[{i-2}], S[{i-1}], S[i])$ and $(S[{i-1}], S[i], S[{i+1}])$ consist of pairwise distinct characters\rev{, respectively}.
\rev{Thus, we obtain $S[i - 2] = S[i + 1]$, and by \Cref{obs:swappable}, $S[i - 1]$ and $S[i]$} are swappable, contradicting the minimality of $i$.

Let \rev{$S^*$} be the coloring string obtained from $S$ by swapping $S[i]$ and $S[{i+1}]$.  
Then \rev{$S^*[i-2] = S[{i-2}] = S[i] = S^*[i+1]$}, hence $S^*[i-1]$ and $S^*[i]$ are swappable in $S^*$.  
By the induction hypothesis, the claim holds for \rev{$S^*$}, and hence also for $S$.
\end{proof}

\begin{lemma}\label{lem:reconfigurable_add}
	Let $S$ and $S'$ be coloring strings, and let $(w_1, w_2, w_3)$ and $(w'_1, w'_2, w'_3)$ be two triples of pairwise distinct characters such that $w_1 w_2 w_3 S$ and $w'_1 w'_2 w'_3 S'$ are coloring strings.
	If $S$ and $S'$ are reconfigurable, then $w_1 w_2 w_3 S$ and $w'_1 w'_2 w'_3 S'$ are also reconfigurable.
\end{lemma}

\begin{proof}
We prove that for any two adjacent coloring strings $S_a$ and $S_b$, any two strings of the form $w_1w_2w_3S_a$ and $w'_1w'_2w'_3S_b$, where each of $(w_1, w_2, w_3)$ and $(w'_1, w'_2, w'_3)$ is pairwise distinct, are reconfigurable. 
\rev{By repeatedly applying this claim, we complete the proof of \Cref{lem:reconfigurable_add}}.

Suppose that $S_b$ is obtained from $S_a$ by swapping $S_a[i]$ and $S_a[i+1]$. 
\rev{If $i \ge 2$, this swap does not affect the prefix $w_1w_2w_3$ of $w_1w_2w_3S_a$.}
\rev{Hence}, $w_1w_2w_3S_b$ can be obtained directly from $w_1w_2w_3S_a$ \rev{by performing the same swap on the corresponding positions in $S_a$}.
Moreover, since $w_1w_2w_3S_b$ and $w'_1w'_2w'_3S_b$ share the same suffix \rev{and both have prefixes of length~$3$ consisting of pairwise distinct characters, they can be reconfigured by swapping characters only within their prefixes while keeping the suffix fixed}. 
Thus, $w_1w_2w_3S_a$ and $w'_1w'_2w'_3S_b$ are reconfigurable.

If $i = 1$, we consider two subcases depending on whether $w_3$ is equal to $S_a[2]$.

Suppose that $w_3 \neq S_a[2]$.
\rev{Since the swap of $S_a[1]$ and $S_a[2]$ can be performed independently of the prefix $w_1w_2w_3$, the same argument as in the case where $i \ge 2$ applies.}
Thus, $w_1w_2w_3S_a$ and $w'_1w'_2w'_3S_b$ are reconfigurable.

Suppose {$w_3 = S_a[2]$}.
In this case, we need to move $w_3$ to a different position before applying the swap \rev{of $S_a[1]$ and $S_a[2]$}.
Depending on the characters of $w_1$ and $w_2$, we proceed as follows:
\begin{itemize}
    \item If $w_2 \neq S_a[1]$, swap characters through the following sequence : $w_1w_2w_3S_a \to w_1w_3w_2S_a \to w_1w_3w_2S_b.$
    \item If $w_1 \neq S_a[1]$, swap characters through the following sequence : $w_1w_2w_3S_a \to w_2w_1w_3S_a \to w_2w_3w_1S_a \to w_2w_3w_1S_b.$
\end{itemize}
In both cases, the resulting string (either $w_1w_3w_2S_b$ or $w_2w_3w_1S_b$) can be reconfigured to $w'_1w'_2w'_3S_b$ since their prefixes are both formed by three distinct characters and share the same suffix $S_b$.
Thus, $w_1w_2w_3S_a$ and $w'_1w'_2w'_3S_b$ are reconfigurable.

In all cases, we conclude that $w_1w_2w_3S_a$ and $w'_1w'_2w'_3S_b$ are reconfigurable.
\rev{This completes the proof of \Cref{lem:reconfigurable_add}.}
\end{proof}

The following is the main theorem in this subsection.
\begin{theorem}\label{lem:iq_iff_yes}
	Let $(P, f_s, f_t)$ be a valid instance of $\threeCRCS$ such that $P$ is a path of $n$-vertices, and let $S$ and $S'$ be the coloring strings corresponding to $f_s$ and $f_t$, respectively.
	Then, $\inv(S) = \inv(S')$ if and only if $S$ and $S'$ are reconfigurable.
\end{theorem}
\begin{proof}
	The ``if'' direction follows directly from \Cref{lem:same_iq_CS}.
	Hence, we now prove the ``only-if'' direction.
	We show that $S$ and $S'$ are reconfigurable if $\inv(S) = \inv(S')$, by induction on the length of $S$ (and $S'$).

	For the base case where $|S| = |S'| \in \{0,1,2\}$, we always have $\inv(S) = \inv(S') = \mathsf{NIL}$, and it is clear that $S$ and $S'$ are reconfigurable.

	For the induction step, assume that the claim holds for all coloring strings shorter than $n$.
	Now consider $n = |S| = |S'| \geq 3$.
	Since $\inv(S) = \inv(S')$, either both $S$ and $S'$ are rigid, or both are non-rigid.
	If both are rigid, then $S = \inv(S) = \inv(S') = S'$, so $S$ and $S'$ are identical and thus trivially reconfigurable.

	Suppose $S$ and $S'$ are not rigid.
	By \Cref{lem:collect_top}, there exist coloring strings $S_x$ and $S_y$ such that the first three characters of each string are pairwise distinct, and $S$ is reconfigurable to $S_x$, while $S'$ is reconfigurable to $S_y$.
    This leads to $\inv(S_x) = \inv(S)$ and $\inv(S_y) = \inv(S')$ by \Cref{lem:same_iq_CS}.
	Note that $S_x$ and $S_y$ can be contracted to $S_x[4,n]$ and $S_y[4,n]$, respectively.
	Since $\inv(S_x[4,n]) = \inv(S_x) = \inv(S) = \inv(S') = \inv(S_y) = \inv(S_y[4,n])$, by the induction hypothesis, $S_x[4,n]$ and $S_y[4,n]$ are reconfigurable.

	Applying \Cref{lem:reconfigurable_add}, it follows that $S_x$ and $S_y$ are also reconfigurable.
	Consequently, by the transitivity of reconfigurability, $S$ and $S'$ are reconfigurable.
\end{proof}

\subsection{Cographs}
\label{subsec:cograph}

We begin by defining the class of \emph{cographs}, also known as $P_4$-free graphs~\cite{Cograph:Corneil81}.
For two graphs $G_1 = (V_1, E_1)$ and $G_2 = (V_2, E_2)$, their \emph{disjoint union} is defined as $G_1 \oplus G_2 = (V_1 \cup V_2, E_1 \cup E_2)$, and their \emph{join} is defined as $G_1 \otimes G_2 = (V_1 \cup V_2, E_1 \cup E_2 \cup \{v_1v_2 \mid v_1 \in V_1, v_2 \in V_2\})$.
A graph $G = (V, E)$ is a \emph{cograph} if it can be constructed recursively according to the following rules:
\begin{enumerate}
	\item A graph consisting of a single vertex is a cograph.
	\item If $G_1$ and $G_2$ are cographs, then their disjoint union $G_1 \oplus G_2$ is also a cograph.
	\item If $G_1$ and $G_2$ are cographs, then their join $G_1 \otimes G_2$ is also a cograph.
\end{enumerate}

This inductive definition yields a canonical tree representation called a \emph{cotree}, where each leaf corresponds to a vertex of the graph, and each internal node is labeled as either a disjoint union node ($\oplus$) or a join node ($\otimes$).
Note that, by the definition of cographs, a cotree is a binary tree.
For a cograph $G$, let $T_G$ denote a cotree corresponding to $G$, and for a node $x \in V(T_G)$, let $G_x$ denote the subgraph of $G$ induced by the leaves of the subtree of $T_G$ rooted at $x$.
It is known that a cotree of a given cograph can be computed in linear time~\cite{Cograph:Tedder2008}.

In this subsection, we provide a polynomial-time algorithm for {\CRCS} on cographs.

\begin{theorem}
	\label{thm:cographs}
	{\CRCS} can be solved in polynomial time for cographs.
\end{theorem}

\subsubsection{Extended \texorpdfstring{$k$}{k}-Colorings}
To describe our algorithm, we introduce the notion of \emph{extended colorings}.
Let $k$ be a positive integer.
An \emph{extended $k$-coloring} of a graph $G$ is a map $f \colon V(G) \to [k] \cup \{\ast\}$ such that for every edge $uv \in E(G)$, either $f(u) \ne f(v)$ or $f(u) = f(v) = \ast$.
That is, we allow a special flexible color $\ast$ in addition to the usual color set $[k]$, and adjacent vertices are permitted to share the color $\ast$.
Given an extended $k$-coloring $f$ of $G$, we define the set of \emph{swappable colors} with respect to $f$ as
$\swappable{f} = \{c \in [k]\cup\{\ast\} \mid |f^{-1}(c)| = 1 \text{ or } (c = \ast  \text { and } f^{-1}(c) \neq \emptyset) \}$.
Namely, a color $c$ is swappable if it is used by exactly one vertex in $G$ under $f$, or if $c = \ast$ and at least one vertex is assigned the color $\ast$ under $f$.

Our polynomial-time algorithm solves a generalized version of {\CRCS}, which we call the \prb{Extended $k$-Coloring Reconfiguration} ({\ECRCS}) problem.
In {\ECRCS}, we are given a graph $G$ and two extended $k$-colorings $f_s$ and $f_t$ of $G$, and the goal is to determine whether there exists a sequence of extended $k$-colorings that transforms $f_s$ into $f_t$ via color swaps.
We use the terminology for {\ECRCS} in the same way as for {\CRCS}.
Note that {\ECRCS} generalizes {\CRCS} in the following sense: for any instance $(G, k, f_s, f_t)$ of {\CRCS}, it holds that $(G, k, f_s, f_t)$ is a yes-instance of {\CRCS} if and only if it is a yes-instance of {\ECRCS}.
Furthermore, for any valid instance $(G, k, f_s, f_t)$ of {\ECRCS}, the initial and target colorings $f_s$ and $f_t$ must have the same set of swappable colors, i.e., $\swappable{f_s} = \swappable{f_t}$.

\subsubsection{Polynomial-time Algorithm for Cographs}
We now \rev{present a polynomial-time algorithm for {\ECRCS} on cographs}.
Our approach is inspired by the polynomial-time algorithm for \prb{Token Sliding} on cographs~\cite{DdTS:KaminskiMM12}.
Indeed, an extended $1$-coloring of a graph $G$ naturally corresponds to an independent set of $G$, and hence our algorithm generalizes their result.

We briefly describe how our algorithm works.
The algorithm \rev{recursively processes} the cotree $T_G$ of the input cograph $G$ from the root to the leaves, solving the {\ECRCS} instance associated with each node $x$ of $T_G$.
\rev{
If $x$ is a leaf node, then $G_x$ consists of a single vertex.
In this case, we simply check whether this vertex receives the same color in both the initial and target extended $k$-colorings.}

\rev{
If $x$ is an internal node with children $x_1$ and $x_2$, the procedure depends on the type of $x$.
When $x$ is a union node, the subgraphs $G_{x_1}$ and $G_{x_2}$ are disconnected, hence the corresponding subproblems can be solved independently.}

\rev{When $x$ is a join node, we first compute the sets of swappable colors 
$\swappable{f_s^1}, \swappable{f_t^1}, \swappable{f_s^2}, \swappable{f_t^2}$ for the restrictions of $f_s$ and $f_t$ to $G_{x_1}$ and $G_{x_2}$, respectively.
These sets determine whether color swaps can occur between vertices in $G_{x_1}$ and $G_{x_2}$ (\Cref{obs:swappable_cograph}), and we distinguish two cases.
If no color swap is possible between $G_{x_1}$ and $G_{x_2}$, the problem again decomposes into independent subproblems on $G_{x_1}$ and $G_{x_2}$ (\Cref{lem:cograph_join_empty}). 
Otherwise, we identify the vertices involved in swaps between $G_{x_1}$ and $G_{x_2}$ and modify the colorings by assigning the special color $\ast$ to those vertices (\Cref{lem:cograph_join_nonempty}). 
This transformation reduces the problem to independent subproblems on $G_{x_1}$ and $G_{x_2}$.
}

We now turn to the formal description of our algorithm.
Let $x$ be an internal node of the cotree of $G$ with two children $x_1$ and $x_2$.
For a $k$-coloring $f$ of $G$, let $f^1$ and $f^2$ denote the restrictions of $f$ to $V(G_{x_1})$ and $V(G_{x_2})$, respectively.
Suppose that we are given an instance $(G, k, f_s, f_t)$ of {\ECRCS}.

\subparagraph{Leaf node.}
Let $G$ be a cograph consisting of a single vertex $v$.
Observe that $(G, k, f_s, f_t)$ is a yes-instance of {\ECRCS} if and only if $f_s(v) = f_t(v)$.

\subparagraph{Union node.}
We consider the case where the root node of $T_G$ is a union node with children $x_1$ and $x_2$.
Since a union operation introduces no edges between $V(G_{x_1})$ and $V(G_{x_2})$, swaps involving vertices in $V(G_{x_1})$ and those in $V(G_{x_2})$ occur independently.
Therefore, $(G, k, f_s, f_t)$ is a yes-instance of {\ECRCS} if and only if both $(G_{x_1}, k, f_s^1, f_t^1)$ and $(G_{x_2}, k, f_s^2, f_t^2)$ are yes-instances of {\ECRCS}.

\subparagraph{Join node.}
We now consider the case where the root node of $T_G$ is a join node with children $x_1$ and $x_2$.
Note that since $x$ is a join node, all vertices in $V(G_{x_1})$ are adjacent to all vertices in $V(G_{x_2})$.
Consequently, for any $k$-coloring of $G$, no color can appear in both $f^1$ and $f^2$.

We perform a case analysis based on whether any of the swappable color sets $\swappable{f_s^1}$, $\swappable{f_s^2}$, $\swappable{f_t^1}$, or $\swappable{f_t^2}$ is empty.
We begin with the following observation.

\begin{observation}
	\label{obs:swappable_cograph}
	Let $(G, k, f_s, f_t)$ be a yes-instance of {\ECRCS}, where $G$ is a cograph and the root node of $T_G$ is a join node with children $x_1$ and $x_2$.
	Then, the following statements hold:
	\begin{enumerate}
		\item If at least one of $\swappable{f_s^1}$ or $\swappable{f_s^2}$ is empty, then no color swap between a vertex in $V(G_{x_1})$ and a vertex in $V(G_{x_2})$ occurs in any reconfiguration sequence from $f_s$ to $f_t$.
		\item $\swappable{f_s^1} = \emptyset$ if and only if $\swappable{f_t^1} = \emptyset$; similarly, $\swappable{f_s^2} = \emptyset$ if and only if $\swappable{f_t^2} = \emptyset$.
	\end{enumerate}
\end{observation}

\begin{proof}
    We first prove the first claim.  
    By symmetry, we may assume without loss of generality that $\swappable{f_s^1} = \emptyset$.  
    Suppose, for contradiction, that there exists a reconfiguration sequence $f_0, f_1, \ldots, f_{\ell}$ from $f_0 = f_s$ to $f_\ell = f_t$, in which a color swap between a vertex in $V(G_{x_1})$ and a vertex in $V(G_{x_2})$ \rev{occurs}.
    Let $i^* \in [\ell]$ be the smallest index such that $f_{i^*}$ is obtained from $f_{i^*-1}$ by swapping the colors of $u \in V(G_{x_1})$ and $v \in V(G_{x_2})$.  
    By the minimality of $i^*$, each color swap in the sequence $f_0, f_1, \ldots, f_{i^*-1}$ occurs entirely within \rev{either} $V(G_{x_1})$ or $V(G_{x_2})$.  
    Therefore, $\swappable{f^1_{i^*-1}} = \swappable{f^1_{i^*-2}} = \cdots = \swappable{f^1_0} = \emptyset$, and in particular, $f_{i^*-1}(u) \neq \ast$.  

    Since $f_{i^*-1}(u) \notin \swappable{f^1_{i^*-1}}$, there must exist another vertex $u' \in V(G_{x_1})$, with $u' \neq u$, such that $f_{i^*-1}(u') = f_{i^*-1}(u)$. 
    \rev{Recall that the root is a join node, and thus $u'$ and $v$ are adjacent in $G$.} 
    Consequently, \rev{after swapping the colors of $u$ and $v$,} both $u'$ and $v$ are assigned the same color \rev{(other than $\ast$)} in $f_{i^*}$, contradicting the fact that $f_{i^*}$ is an extended proper $k$-coloring of $G$.  
    This implies that no swap between $V(G_{x_1})$ and $V(G_{x_2})$ can occur in any reconfiguration sequence from $f_s$ to $f_t$.
    
    The second claim follows directly from this argument:  
    since $\swappable{f^1_0} = \emptyset$ and remains unchanged throughout the \rev{reconfiguration} sequence, we have $\rev{\emptyset = \swappable{f^1_s}} = \swappable{f^1_0} = \swappable{f^1_1} = \cdots = \swappable{f^1_\ell} = \swappable{f_t^1}$.  
\end{proof}

We first consider the case where at least one of $\swappable{f_s^1}$ or $\swappable{f_s^2}$ is empty.

\begin{lemma}
	\label{lem:cograph_join_empty}
	Let $(G, k, f_s, f_t)$ be a valid instance of {\ECRCS}, where $G$ is a cograph and the root node of $T_G$ is a join node with children $x_1$ and $x_2$.
	Suppose that at least one of $\swappable{f_s^1}$ or $\swappable{f_s^2}$ is empty.
	Then, $(G, k, f_s, f_t)$ is a yes-instance of {\ECRCS} if and only if both $(G_{x_1}, k, f^{1}_{s}, f^{1}_{t})$ and $(G_{x_2}, k, f^{2}_{s}, f^{2}_{t})$ are yes-instances of {\ECRCS}.
\end{lemma}

\begin{proof}
    Without loss of generality, suppose $\swappable{f_s^1} = \emptyset$.
    \rev{The other case can be proved symmetrically.}
    
    We first prove the ``only if'' direction.
    Assume that $(G, k, f_s, f_t)$ is a yes-instance, and let $\mathcal{R}$ be a reconfiguration sequence from $f_s$ to $f_t$.
    By \Cref{obs:swappable_cograph}, no swap in $\mathcal{R}$ occurs between a vertex in $V(G_{x_1})$ and a vertex in $V(G_{x_2})$.
    Thus, each step in $\mathcal{R}$ modifies \rev{the colors of} vertices \rev{only} in $V(G_{x_1})$ or only in $V(G_{x_2})$.
    
    Recall that for each $j \in \{1, 2\}$, $f^j_s$ (resp., $f^j_t$) denotes the restriction of $f_s$ (resp., $f_t$) to $V(G_{x_j})$.
    Thus, by restricting $\mathcal{R}$ to $V(G_{x_j})$ and deleting any consecutive identical extended $k$-colorings, we obtain a reconfiguration sequence from $f^j_s$ to $f^j_t$ in $G_{x_j}$.
    Thus, both $(G_{x_1}, k, f^{1}_{s}, f^{1}_{t})$ and $(G_{x_2}, k, f^{2}_{s}, f^{2}_{t})$ are yes-instances.
    
    We now prove the ``if'' direction.  
    Suppose that both $(G_{x_1}, k, f^{1}_{s}, f^{1}_{t})$ and $(G_{x_2}, k, f^{2}_{s}, f^{2}_{t})$ are yes-instances, and let $\mathcal{R}_1$ and $\mathcal{R}_2$ be corresponding reconfiguration sequences in $G_{x_1}$ and $G_{x_2}$, respectively.

    We construct a reconfiguration sequence $\mathcal{R}$ from $f_s$ to $f_t$ in $G$ as follows:
    \begin{enumerate}
        \item Initialize $f \coloneqq f_s$.
        \item For each swap in $\mathcal{R}_1$ (processed in order), apply the corresponding swap to $f$, and update $f$ accordingly.
        \item Then, for each swap in $\mathcal{R}_2$ (processed in order), apply the corresponding swap to $f$, and update $f$ accordingly.
    \end{enumerate}
    \rev{Recall that each step in $\mathcal{R}_1$ (resp., $\mathcal{R}_2$) swaps colors only within $V(G_{x_1})$ (resp., $V(G_{x_2})$).}
    By the definition of $f_t^1$ and $f_t^2$, the final coloring of this process is $f_t$.
    Since no color appears in both $f_s^1$ and $f_s^2$, and all vertices in $V(G_{x_1})$ are adjacent to all vertices in $V(G_{x_2})$, 
    every intermediate map in $\mathcal{R}$ remains a proper extended $k$-coloring of $G$.

    Therefore, $f_s$ is reconfigurable to $f_t$ under $\CS$, and hence $(G, k, f_s, f_t)$ is a yes-instance.
\end{proof}

We then consider the remaining case.
Let $(G, k, f_s, f_t)$ be a valid instance of {\ECRCS}, where $G$ is a cograph and the root node of $T_G$ is a join node with children $x_1$ and $x_2$.
For each $j = 1, 2$, we define extended $k$-colorings $f^j_{s\ast}, f^j_{t\ast} \colon V(G_{x_j}) \to [k] \cup \{\ast\}$ as follows:
\begin{align*}
	f^j_{s\ast}(v) =
	\begin{cases}
		\ast     & \text{if } f_s^j(v) \in \swappable{f_s}, \\
		f_s^j(v) & \text{otherwise}
	\end{cases}, \quad
	f^j_{t\ast}(v) =
	\begin{cases}
		\ast     & \text{if } f_t^j(v) \in \swappable{f_t}, \\
		f_t^j(v) & \text{otherwise}
	\end{cases}.
\end{align*}
By construction, both $f^j_{s\ast}$ and $f^j_{t\ast}$ are extended $k$-colorings of $G_{x_j}$.

\begin{lemma}
	\label{lem:cograph_join_nonempty}
	Let $(G, k, f_s, f_t)$ be a valid instance of {\ECRCS}, where $G$ is a cograph and the root node of $T_G$ is a join node with children $x_1$ and $x_2$.
	Suppose that both $\swappable{f_s^1}$ and $\swappable{f_s^2}$ are non-empty.
	Then, $(G, k, f_s, f_t)$ is a yes-instance of {\ECRCS} if and only if both $(G_{x_1}, k, f^1_{s\ast}, f^1_{t\ast})$ and $(G_{x_2}, k, f^2_{s\ast}, f^2_{t\ast})$ are yes-instances of {\ECRCS}.
\end{lemma}

\begin{proof}
    Let $I_1 = (G_{x_1}, k, f^1_{s\ast}, f^1_{t\ast})$ and $I_2 = (G_{x_2}, k, f^2_{s\ast}, f^2_{t\ast})$.
    
    We first prove the ``only if'' direction.
    Suppose that $(G, k, f_s, f_t)$ is a yes-instance, and let $f_0, f_1, \ldots, f_\ell$ be a reconfiguration sequence from $f_0 = f_s$ to $f_\ell = f_t$.
    
    We construct reconfiguration sequences $\mathcal{R}_1$ and $\mathcal{R}_2$ for $I_1$ and $I_2$ starting from $f^1_{s\ast}$ and $f^2_{s\ast}$, respectively, as follows.
    For each $i \in [\ell]$, suppose that $f_i$ is obtained from $f_{i-1}$ by swapping the colors of the endpoints of an edge $uv$:
    \begin{itemize}
        \item If $u, v \in V(G_{x_1})$, simulate this swap in $\mathcal{R}_1$.
        \item If $u, v \in V(G_{x_2})$, simulate this swap in $\mathcal{R}_2$.
        \item If $u \in V(G_{x_1})$ and $v \in V(G_{x_2})$, do nothing.
    \end{itemize}
    Since no swap across $V(G_{x_1})$ and $V(G_{x_2})$ is simulated, each intermediate coloring in $\mathcal{R}_1$ and $\mathcal{R}_2$ remains a proper extended $k$-coloring of $G_{x_1}$ and $G_{x_2}$, respectively.
    By construction, the final colorings of $\mathcal{R}_1$ and $\mathcal{R}_2$ are $f^1_{t\ast}$ and $f^2_{t\ast}$, respectively.
    Therefore, both $I_1$ and $I_2$ are yes-instances of {\ECRCS}.

    We now prove the ``if'' direction.
    Suppose that both $I_1$ and $I_2$ are yes-instances.
    Let $\mathcal{R}_1$ and $\mathcal{R}_2$ be reconfiguration sequences from $f^1_{s\ast}$ to $f^1_{t\ast}$ in $G_{x_1}$ and from $f^2_{s\ast}$ to $f^2_{t\ast}$ in $G_{x_2}$, respectively.

    We construct a reconfiguration sequence $\mathcal{R}$ from $f_s$ to an extended $k$-coloring $g$ of $G$ as follows:
    \begin{enumerate}
        \item Initialize $f \coloneqq f_s$.
        \item For each swap in $\mathcal{R}_1$ (processed in order), apply the corresponding swap to $f$, and update $f$ accordingly.
        \item Then, for each swap in $\mathcal{R}_2$ (processed in order), apply the corresponding swap to $f$, and update $f$ accordingly.
    \end{enumerate}
    Let $g$ denote the resulting map.
    Since no swap between a vertex in $V(G_{x_1})$ and a vertex in $V(G_{x_2})$ is performed, every intermediate map in this process remains a proper extended $k$-coloring of $G$.
    Thus, \rev{$\mathcal{R}$ is indeed a reconfiguration sequence; that is,} $f_s$ is reconfigurable to $g$.

    Next, we construct a reconfiguration sequence from $g$ to $f_t$.
    Let $X \coloneqq \{v \in V(G) \mid g(v) \in \swappable{g}\}$.
    By construction, we have $g(v) = f_t(v)$ for all $v \in V(G) \setminus X$.
    Therefore, we only need to modify the colors of $g$ on $X$.
    Since $f_s$ is reconfigurable to $g$, we have $\swappable{f_s} = \swappable{g}$.
    As $(G, k, f_s, f_t)$ is a valid instance, we also have $\swappable{f_s} = \swappable{f_t}$, and thus $\swappable{g} = \swappable{f_t}$.

    \rev{Recall that $g^1$ and $g^2$ are the restrictions of $g$ to $V(G_{x_1})$ and $V(G_{x_2})$, respectively.}
    Note that both $\swappable{g^1}$ and $\swappable{g^2}$ are non-empty by \rev{$S_{g^i} = S_{f^i_s} \neq \emptyset$ for each $i \in \{1, 2\}$} and \Cref{obs:swappable_cograph}.
    Since $x$ is a join node, the subgraph of $G$ induced by $X = \swappable{g^1} \cup \swappable{g^2}$ is connected.
    Therefore, using a polynomial-time algorithm for \prb{Colored Token Swapping} on connected graphs~\cite{color_token:YamanakaHKKOSUU18} \rev{(which implies that any two bijective $|V(H)|$-colorings of a connected graph $H$ are reconfigurable)}, we can construct a reconfiguration sequence from $g$ to $f_t$ that swaps colors \rev{of only vertices in} $X$.
    Combined with the fact that $f_s$ is reconfigurable to $g$, it follows that $f_s$ is reconfigurable to $f_t$ under $\CS$.
    Hence, $(G, k, f_s, f_t)$ is a yes-instance of {\ECRCS}.
\end{proof}

\subparagraph{Algorithm.}
The arguments for leaf, union, and join nodes naturally lead to a recursive algorithm for {\ECRCS} on cographs.
Since the set of swappable colors at each node of $T_G$ and the corresponding extended colorings can be constructed in polynomial time, the overall algorithm runs in polynomial time.
This concludes the proof of \cref{thm:cographs}.

\subsection{Split Graphs}\label{subsec:split}
Recall that {\CRCS} is $\PSPACE$-complete on split graphs, as shown in \Cref{thm:CRCS_PSPACEcomp_split}.
In this subsection, we show a contrasting result: $\kCRCS$ is solvable in polynomial time on split graphs when the number of colors $k$ is a constant.

\begin{theorem}\label{thm:Kernelizeforsplit}
	{\CRCS} on split graphs admits a kernel with at most $k + 2k2^k$ vertices, where $k$ is the number of colors of an input.
\end{theorem}

\rev{Recall that} a \emph{kernelization algorithm} \rev{(or simply a \emph{kernel})} for a parameterized problem~$Q$ 
\rev{is a polynomial-time algorithm that transforms any instance $(I, p)$ of~$Q$ 
into an equivalent instance $(I', p')$ of~$Q$ 
such that the size of $(I', p')$ is bounded by a computable function $h(p)$ that depends only on the parameter~$p$.}
The function $h$ of $p$ is referred to as the \emph{size} of the kernel.
In the following, we say that a reduction rule is \emph{safe} if the instance $\rev{(I', p')}$ obtained by applying the rule is equivalent to the original instance $\rev{(I, p)}$.

Let $(G, f_s, f_t, \rev{k})$ be an instance of $\CRCS$, where $G$ is a split graph whose vertex set $V(G) = C \cup I$ consists of a clique $C$ and an independent set $I$.
We begin by introducing a reduction rule based on the following observation.
\rev{Consider} two vertices $u, v \in I$ such that $N(u) = N(v)$, \rev{and let $f$ be a proper $k$-coloring of $G$ satisfying $f(u) = f(v)$.}
\rev{Observe that during any reconfiguration sequence starting from~$f$, no color swap can involve either $u$ or~$v$.}
We call such vertices \emph{frozen} with respect to $f$.
This observation \rev{further implies that if $f'$ is a proper $k$-coloring} adjacent to $f$ under~$\CS$, then every vertex frozen with respect to~$f$ remains frozen with respect to~$f'$. 
Consequently, \rev{if there exists a frozen vertex~$u$ with respect to both $f_s$ and $f_t$ such that $f_s(u) \neq f_t(u)$}, 
then there exists no reconfiguration sequence from $f_s$ to $f_t$.
\rev{This motivates the following reduction rule.}

\begin{redrule}
    \label{red:split_rule1}
    Let $u, v \in I$ be vertices such that $N(u) = N(v)$ and $f_s(u) = f_s(v)$.
    If $f_s(u) \neq f_t(u)$ or $f_s(v) \neq f_t(v)$ holds, then return \rev{\texttt{NO}}.
\end{redrule}

\rev{Next,} let $u, v, w \in I$ be three vertices such that $N(u) = N(v) = N(w)$ \rev{and $f_s(u) = f_s(v) = f_s(w)$}.
After applying \Cref{red:split_rule1}, we then have $f_t(u) = f_t(v) = f_t(w)$.
\rev{In this case, the vertex~$w$ can be safely deleted from~$G$, leading to the following reduction rule.}

\begin{redrule}
    \label{red:split_rule2}
    Let $u, v, w \in I$ be three vertices such that $N(u) = N(v) = N(w)$.
    If $f_s(u) = f_s(v) = f_s(w)$, then return an instance $(G',f'_s, f'_t, \rev{k})$, where $G'$ is the graph obtained by removing $w$ from $G$, and $f'_s(x) = f_s(x)$ and $f'_t(x) = f_t(x)$ for all $x \in V(G) \setminus \{w\}$.    
\end{redrule}

\begin{lemma}\label{lem:Rule2_safe}
    \Cref{red:split_rule2} is safe.
\end{lemma}

\begin{proof}
    \rev{Note that $u, v$ and $w$ are frozen with respect to $f_s$ and $f_t$.}
    Since $G$ is a split graph, $G'$ is also a split graph with the vertex set $C \cup (I \setminus \{w\})$.
    We now prove that $(G, f_s, f_t, \rev{k})$ is a yes-instance if and only if $(G', f'_s, f'_t, \rev{k})$ is a yes-instance.

    Suppose that $(G, f_s, f_t, \rev{k})$ is a yes-instance.
    Then there exists a reconfiguration sequence $f_0, f_1, \ldots, f_\ell$ such that $f_0 = f_s$ and $f_\ell = f_t$, and $f_i$ \rev{for each $i \in [\ell]$} is a proper $k$-coloring of $G$ adjacent to $f_{i-1}$ under $\CS$.
    We construct the corresponding sequence $f'_0, f'_1, \ldots, f'_\ell$ on $G'$ by restricting each $f_i$ to \rev{$V(G')= V(G) \setminus \{w\}$}; that is, for every $i \in \{0, 1, \ldots, \ell\}$ and $x \in V(G')$, we define $f'_i(x) = f_i(x)$.
    \rev{Since $w$ is frozen and \Cref{red:split_rule1} cannot be applied, we have $f_i(w) = f_s(w) = f_t(w)$ for all $i \in \{0, 1, \ldots, \ell\}$.}
    \rev{Moreover, for each $i\in[\ell]$, there exists exactly one edge $xy \in E(G')$ such that $f'_{i-1}(x) = f'_i(y)$ and $f'_{i-1}(y) = f'_i(x)$.}
    This implies that $f'_{i-1}$ and $f'_i$ are adjacent under $\CS$.
    Therefore, $f'_0, f'_1, \ldots, f'_\ell$ forms a reconfiguration sequence from $f'_s$ to $f'_t$ in $G'$, and \rev{thus} $(G', f'_s, f'_t, \rev{k})$ is a yes-instance.

    Conversely, suppose that $(G', f'_s, f'_t, \rev{k})$ is a yes-instance.
    Then there exists a reconfiguration sequence $f'_0, f'_1, \ldots, f'_\ell$ such that $f'_0 = f'_s$ and $f'_{\ell} = f'_t$ in $G'$, where each $f'_i$ is adjacent to $f'_{i-1}$ under $\CS$.
    For each $i\in \{0, 1, \ldots, \ell\}$, we construct a $k$-coloring $f_i$ of $G$ \rev{as follows}:
    \begin{align*}
        f_i(x)=\begin{cases}
                f_s(w) & \text{if $x=w$,} \\
                f'_i(x) & \text{otherwise.} 
              \end{cases}
    \end{align*}

    Since no color swap involves either $u$ or $v$ during the \rev{reconfiguration} sequence $f'_0, f'_1, \ldots, f'_\ell$, \rev{the assumption of \Cref{red:split_rule2} implies} that $f'_i(u) = f'_i(v) = f_s(w)$ for all $i \in [\ell]$.
    Thus, for any $y \in N(w)$ and $i \in \{0, 1, \ldots, \ell\}$, we have $f'_i(y) \neq f_i(w)$, implying that $f_i$ is a proper coloring of $G$.
    Moreover, for every $i \in [\ell]$, \rev{the colorings} $f'_{i -1}$ and $f'_{i}$ differ exactly \rev{on two vertices $x$ and $y$ in $G'$ such that $xy \in E(G')$ and $x, y \neq w$}. 
    Thus, $f_{i-1}$ and $f_{i}$ also differ only \rev{on} $x$ and $y$ \rev{corresponding to the same edge $xy \in E(G)$}, 
    and hence \rev{they} are adjacent under $\CS$ in $G$.
    Therefore, the sequence $f_0, f_1, \ldots, f_\ell$ forms indeed a reconfiguration sequence between $f_s$ and $f_t$ in $G$, and $(G, k, f_s, f_t)$ is a yes-instance.
\end{proof}

\rev{We then give the proof of \Cref{thm:Kernelizeforsplit}}.

\begin{proof}[Proof of \Cref{thm:Kernelizeforsplit}]    
    Our kernelization algorithm \rev{proceeds} as follows.
    Given an instance \rev{$(G, f_s, f_t, k)$} of {\CRCS}, \rev{we} apply \Cref{red:split_rule1,red:split_rule2} exhaustively, \rev{and let $(G', f'_s, f'_t, k)$ denote the resulting instance}. 
    \rev{Observe that both reduction rules can be done in polynomial time, and the resulting graph $G'$ is still a split graph.}
    \rev{We write $V(G') = C' \cup I'$, where $C'$ is a clique and $I'$ is an independent set of~$G'$.}

    Since $C'$ is a clique and every vertex in $C'$ must receive a distinct color in any proper $k$-coloring of $G'$, we have $|C'| \le k$.

    Next, \rev{we bound the size of $I'$.}
    Since each vertex \rev{in $I'$} is adjacent only to vertices in $C'$, and hence its neighborhood is a subset of $C'$.
    Thus, the number of distinct neighborhoods among the vertices in $I'$ is at most $2^{|C|} \leq 2^k$.

    \rev{Furthermore, for each subset $N' \subseteq C'$ and each color $c \in [k]$, there are at most two vertices in $I'$ with neighborhood $N'$ and assigned color $c$ in $f_s$; otherwise \Cref{red:split_rule2} would apply.}
    Therefore, the total number of vertices in $I'$ is at most $2k \cdot 2^k$.

    \rev{Combining these bounds, the total number of vertices in $G'$ is at most $|V(G')| = |C'| + |I'| \le k + 2k \cdot 2^k$}.
    \rev{Thus, our algorithm outputs a kernel of size $k + 2k \cdot 2^k$ and correctly reduces the instance $(G, f_s, f_t, k)$.}
\end{proof}

\Cref{thm:Kernelizeforsplit} immediately leads to the following \Cref{col:kCRCSpolynomial}.

\begin{corollary}\label{col:kCRCSpolynomial}
	{\kCRCS} can be solved in polynomial time for split graphs.
\end{corollary}

\section{Concluding Remarks}

An intriguing direction is to determine the complexity of {\CRCS} on graph classes that include paths.
For example, does {\CRCS} admit a polynomial-time algorithm on trees or on caterpillars?
Another natural case is interval graphs, which are a subclass of chordal graphs where $\kCRCS$ has been shown to be $\PSPACE$-complete (\Cref{thm:kCRCS_PSPACEcomp_chordal}).
Since our algorithm for paths (\Cref{thm:path_linear}) relies heavily on both the path structure and the restriction $k=3$, it seems difficult to generalize our approach directly to broader classes.
On the other hand, \prb{Token Sliding} is known to be solvable in polynomial time on both trees~\cite{DemaineDFHIOOUY15} and interval graphs~\cite{BonamyB17,BrianskiFHM21}.
Thus, as in the case of cographs (\Cref{thm:cographs}), it is conceivable that algorithms for {\CRCS} could be derived from existing algorithms for \prb{Token Sliding}.



\bibliography{references}

\end{document}